\documentclass{article}

\usepackage[utf8]{inputenc} % allow utf-8 input
\usepackage{hyperref}       % hyperlinks
\usepackage{url}            % simple URL typesetting
\usepackage{booktabs}       % professional-quality tables
\usepackage{amsfonts}       % blackboard math symbols
\usepackage{nicefrac}       % compact symbols for 1/2, etc.
\usepackage{microtype}      % microtypography
\usepackage{xcolor}         % colors
\usepackage{fullpage}
\usepackage{authblk}
\usepackage{hyperref,url}
\usepackage{sidecap}
\sidecaptionvpos{figure}{c}

\title{Numerical Composition of Differential Privacy\footnote{Author ordering is alphabetical. Code is available at \url{https://github.com/microsoft/prv_accountant}.}}

\author[1]{Sivakanth Gopi}
\author[2]{Yin Tat Lee}
\author[1]{Lukas Wutschitz}
\affil[1]{Microsoft\\
\texttt{\{sigopi,lukas.wutschitz\}@microsoft.com}}
\affil[2]{University of Washington\\
\texttt{yintat@uw.edu}}
\date{}

\usepackage{amsthm,thmtools}
\usepackage{thm-restate}
\usepackage{amsmath,amssymb}
\usepackage{graphicx,subcaption}
\usepackage{float}
\usepackage{wrapfig}
\usepackage[ruled,vlined]{algorithm2e}

\usepackage{placeins}
\newcommand{\R}{\mathbb{R}} % REALS
\newcommand{\Z}{\mathbb{Z}} % INTEGERS
\newcommand{\cA}{\mathcal A}

\newcommand{\cM}{\mathcal M}
\newcommand{\cN}{\mathcal N}

\newcommand{\set}[1]{\{#1\}}

\renewcommand{\epsilon}{\varepsilon}
\newcommand{\eps}{\epsilon}

  \newcommand{\beq}{\begin{equation}}
  \newcommand{\eeq}{\end{equation}}
  \newcommand{\beqn}{\begin{equation*}}
  \newcommand{\eeqn}{\end{equation*}}
  \newcommand{\beqr}{\begin{eqnarray}}
  \newcommand{\eeqr}{\end{eqnarray}}
  \newcommand{\beqrn}{\begin{eqnarray*}}
  \newcommand{\eeqrn}{\end{eqnarray*}}
  \newcommand{\bmline}{\begin{multline}}
  \newcommand{\emline}{\end{multline}}
  \newcommand{\bmlinen}{\begin{multline*}}
  \newcommand{\emlinen}{\end{multline*}}

\usepackage{todonotes}

\newcounter{mynotes}
\def\notes{1}
\newcommand{\gnote}[1]{\ifnum\notes=1{{\sf\color{blue} [Gopi: #1]}}\fi}
\newcommand{\ynote}[1]{\ifnum\notes=1{{\sf\color{red} [YinTat: #1]}}\fi}

\declaretheorem[within=section]{theorem}
\declaretheorem[sibling=theorem]{corollary}

\declaretheorem[sibling=theorem]{lemma}

\declaretheorem[sibling=theorem]{definition}

\declaretheorem[sibling=theorem]{proposition}
\declaretheorem[sibling=theorem]{remark}

\newcommand{\Lap}{\mathsf{Lap}}

\DeclareMathOperator{\E}{\mathbb{E}}

\newcommand{\CDF}{\mathrm{CDF}}

\newcommand{\ty}{\widetilde{y}}
\newcommand{\tY}{\widetilde{Y}}

\newcommand{\tX}{\widetilde{X}}

\graphicspath{{Figures/}}

\newcommand{\error}{\mathrm{error}}
\newcommand{\upper}{\mathrm{upper}}
\newcommand{\lp}{\left(}
\newcommand{\lb}{\left[}
\newcommand{\rb}{\right]}
\newcommand{\rp}{\right)}
\newcommand{\bR}{\overline{\R}}

\begin{document}

\maketitle

\begin{abstract}
We give a fast algorithm to optimally compose privacy guarantees of differentially private (DP) algorithms to arbitrary accuracy. Our method is based on the notion of \emph{privacy loss random variables} to quantify the privacy loss of DP algorithms. %This has an intuitive meaning: if a DP algorithm $M$ has privacy random variables $(X,Y)$ (which are real-valued), then inferring membership of any user from the output of $M$ is at least as hard as distinguishing between $X$ and $Y$ from a random sample. Though mathematically equivalent to the $(\epsilon,\delta)$-DP framework of Dwork et al.~\cite{DMNS06} and the recently proposed $f$-DP framework of Dong et al.~\cite{DongRS19}, this notion allows for a much simpler way to compute the privacy loss of \emph{composition} of DP algorithms, compared to the other frameworks. If $M_1$ and $M_2$ are DP algorithms with privacy random variables $(X_1,Y_1)$ and $(X_2,Y_2)$ respectively, then the (adaptive) composition $M_1\circ M_2$ has privacy random variables $(X_1+X_2,Y_1+Y_2)$.
The running time and memory needed for our algorithm to approximate the privacy curve of a DP algorithm composed with itself $k$ times is $\tilde{O}(\sqrt{k})$. This improves over the best prior method by Koskela et al.~\cite{koskela2021computing} which requires $\tilde{\Omega}(k^{1.5})$ running time. We demonstrate the utility of our algorithm by accurately computing the privacy loss of DP-SGD algorithm of Abadi et al.~\cite{Abadi16} and showing that our algorithm speeds up the privacy computations by a few orders of magnitude compared to prior work, while maintaining similar accuracy.
\end{abstract}

% We give a fast algorithm to optimally compose privacy guarantees of differentially private (DP) algorithms to arbitrary accuracy. Our method is based on the notion of privacy loss random variables to quantify the privacy loss of DP algorithms. The running time and memory needed for our algorithm to approximate the privacy curve of a DP algorithm composed with itself $k$ times is $\tilde{O}(\sqrt{k})$. This improves over the best prior method by Koskela et al. (2020) which requires $\tilde{\Omega}(k^{1.5})$ running time. We demonstrate the utility of our algorithm by accurately computing the privacy loss of DP-SGD algorithm of Abadi et al. (2016) and showing that our algorithm speeds up the privacy computations by a few orders of magnitude compared to prior work, while maintaining similar accuracy.
\newpage
\tableofcontents
\newpage

%\ynote{Need to handle the question 3 and 4 in checklist}
\section{Introduction}
\label{sec:intro}
%!TEX root=./MAIN.tex

Differential privacy (DP) introduced by~\cite{DMNS06} provides a provable and quantifiable guarantee of privacy when the results of an algorithm run on private data are made public. Formally, we can define an $(\eps,\delta)$-differentially private algorithm as follows.
\begin{definition}[$(\eps,\delta)$-DP~\cite{DMNS06,dwork2006our}]
	An algorithm $\cM$ is $(\eps,\delta)$-DP if for any two neighboring databases $D,D'$ differing in exactly one user and any subset $S$ of outputs, we have $\Pr[\cM(D)\in S] \le e^\eps\Pr[\cM(D')\in S]+\delta.$
\end{definition}
Intuitively, it says that looking at the outcome of $\cM$, we cannot tell whether it was run on $D$ or $D'$. Hence an adversary cannot infer the existence of any particular user in the input database, and therefore cannot learn any personal data of any particular user. 

DP algorithms have an important property called \emph{composition}. Suppose $M_1$ and $M_2$ are DP algorithms and say $M(D)=(M_1(D),M_2(D))$, i.e., $M$ runs both the algorithms on $D$ and outputs their results. Then $M$ is also a DP algorithm. 

\begin{proposition}[Simple composition \cite{dwork2006our,dwork2009differential}]
\label{prop:simple_composition}
If $M_1$ is $(\eps_1,\delta_1)$-DP and $M_2$ is $(\eps_2,\delta_2)$-DP, then $M(D)=(M_1(D),M_2(D))$ is $(\eps_1+\eps_2,\delta_1+\delta_2)$-DP.
\end{proposition}

This also holds under \emph{adaptive composition} (denoted by $M=M_2\circ M_1$), where $M_2$ can look at both the database and the output of $M_1$.\footnote{Here $M(D)=(M_1(D),M_2(D,M_1(D)))$.} It turns out that both compositions enjoy much better DP guarantees than this simple composition rule. Let $M$ be an $(\eps,\delta)$-DP algorithm and let $M^{\circ k}$ denote the (adaptive) composition of $M$ with itself $k$ times. The naive composition rule shows that $M^{\circ k}$ is $(k\eps,k\delta)$-DP. This was significantly improved in~\cite{dwork2010boosting}.

\begin{proposition}[Advanced composition \cite{dwork2010boosting,DR14}]
\label{prop:advanced_composition}
	If $M$ is $(\eps,\delta)$-DP, then $M^{\circ k}$ is $(\eps',k\delta+\delta')$-DP where $$\eps'=\eps\sqrt{2k\log\lp \frac{1}{\delta'}\rp}+k\eps(e^\eps -1).$$
\end{proposition}
Note that if $\eps=O\lp \frac{1}{\sqrt{k}}\rp$ and $\delta=o\lp\frac{1}{k}\rp$, then $M^{\circ k}$ satisfies $(O_{\delta'}(1),\delta')$-DP. Using simple composition (Proposition~\ref{prop:simple_composition}), we can only claim that $M^{\circ k}$ is $(O(\sqrt{k}),o(1))$-DP. Thus advanced composition often results in $\sqrt{k}$-factor savings in privacy which is significant in practice. The optimal DP guarantees for $k$-fold composition of an $(\eps,\delta)$-DP algorithm were finally obtained by~\cite{kairouz2015composition}. For composing different algorithms, the situation is more complicated. If $M_1,M_2,\dots,M_k$ are DP algorithms such that $M_i$ is $(\eps_i,\delta_i)$-DP, then it is shown by \cite{murtagh2016complexity} that computing the \emph{exact} DP guarantees for $M=M_1 \circ M_2 \circ \dots \circ M_k$ is \#P-complete. They also give an algorithm to approximate the DP guarantees of $M$ to desired accuracy $\eta$ which runs in 
\begin{equation}
\label{eqn:time_murtaghsalil}
\tilde{O}\lp \frac{k^3 \bar{\eps}(1+\bar{\eps})}{\eta}\rp
\end{equation}
time where $\bar{\eps}=(\sum_{i=1}^k \eps_i)/k.$\footnote{$\eps$ has an additive error of $\eta$ and $\delta$ has a multiplicative error of $\eta$.} If each $\eps_i\approx \frac{1}{\sqrt{k}}$ (so that $M$ will satisfy reasonable privacy guarantees by advanced composition), then the running time is $\tilde{O}(k^{2.5}/\eta).$

In most situations, DP algorithms come with a collection of $(\eps,\delta)$-DP guarantees, i.e., for each value of $\eps$, there exists $\delta$ such that the algorithm is $(\eps,\delta)$-DP.
\begin{definition}[Privacy curve]
A DP algorithm $M$ is said to have privacy curve $\delta:\R \to [0,1]$, if for every $\eps\in \R$, $M$ is $(\eps,\delta(\eps))$-DP.
\end{definition}
For example the privacy curve of a Gaussian mechanism (with sensitivity $1$ and noise scale $\sigma$) is given by $\delta(\eps) = \Phi\lp -\eps \sigma+1/2\sigma \rp - e^\eps \Phi\lp -\eps\sigma-1/2\sigma \rp$ where $\Phi(\cdot)$ is the Gaussian CDF~\cite{BalleW18}.
%$$\delta(\eps) = \Phi\lp -\eps \sigma+\frac{1}{2\sigma} \rp - e^\eps \Phi\lp -\eps\sigma-\frac{1}{2\sigma} \rp$$ where $\Phi(\cdot)$ is the Gaussian CDF~\cite{BalleW18}. %where $\mu=\frac{\sqrt{1}}{\sigma}$ where $k$ is the number of iterations and $\sigma$ is the noise scale.}
Suppose we want to compose several Gaussian mechanisms, which $(\eps,\delta)$-DP guarantee should we choose for each mechanism? Any choice will lead to suboptimal DP guarantees for the final composition. Instead, we need a way to compose the privacy curves directly. This was suggested through the use of privacy region in~\cite{kairouz2015composition} and explicitly studied in the $f$-DP framework of~\cite{DongRS19}. $f$-DP is a dual way (and equivalent) to look at the privacy curve $\delta(\eps).$ 

Independently, an algorithm called \emph{Privacy Buckets} for approximately composing privacy curves using the notion of was initiated in~\cite{meiser2018tight}. This algorithm depends on the notion of \emph{Privacy Loss Random Variable} (PRV)~\cite{DR16}, whose distribution is called \emph{Privacy Loss Distribution} (PLD). For any DP-algorithm, one can associate a PRV and the privacy curve of that algorithm can be easily obtained from the PRV. The really useful property of PRVs is that under adaptive composition, they just add up; the PRV $Y $of the composition $M=M_1\circ M_2 \circ \dots \circ M_k$ is given by $Y=\sum_{i=1}^k Y_i$ where $Y_i$ is the PRV of $M_i.$\footnote{\cite{koskela2020computing} only state this for non-adaptive composition. In this paper we show how to extend this to adaptive composition as well.} Therefore, one can find the distribution of $Y$ by the convolution of the distributions of $Y_1,Y_2,\dots,Y_k$. In an important paper, \cite{koskela2020computing} proposed that one can speed up the convolutions using Fast Fourier Transform (FFT). Explicit error bounds were obtained for the approximation obtained by their algorithm in~\cite{koskela2020computing,KoskelaJPH21,koskela2021computing}. The running time of this algorithm was analyzed in~\cite{koskela2021computing} where it was shown that the privacy curve $\delta_M(\eps)$ of $M=M_1\circ M_2 \circ \dots \circ M_k$ can be computed up to an additive error of $\delta_\error$ in time
\begin{equation}
	\label{eqn:time_Koskela}
\tilde{O}\lp\frac{k^3\bar{\eps}}{\delta_\error}\rp,	
\end{equation}
if each algorithm $M_i$ is satisfies $(\eps_i,0)$-DP and $\bar{\eps}=\frac{1}{k}\sum_{i=1}^k \eps_i$. Assuming that each $\eps_i\approx \frac{1}{\sqrt{k}}$, we get $\tilde{O}(k^{2.5}/\delta_\error)$ running time. Note that this is slightly worse than (\ref{eqn:time_murtaghsalil}), where the denominator $\eta$ is the multiplicative error in $\delta_M$. When composing the same algorithm with itself for $k$ times, the running time can be improved to $\tilde{O}\lp\frac{k^2\bar{\eps}}{\delta_\error}\rp$, which is $\tilde{O}\lp\frac{k^{1.5}}{\delta_\error}\rp$ when $\bar{\eps}=\frac{1}{\sqrt{k}}$.

\paragraph{Moments Accountant and Renyi DP} In an influential paper where they introduce Differentially Private Deep Learning, \cite{Abadi16} proposed a method called the Moments Accountant (MA) for giving an upper bound the privacy curve of a composition of DP algorithms. They applied their method to bound the privacy loss of differentially private Stochastic Gradient Descent (DP-SGD) algorithm which they introduced. Analyzing the privacy loss of DP-SGD involves composing the privacy curve of each iteration of training with itself $k$ times, where $k$ is the total of number of training iterations. Typical values of $k$ range from $1000$ to $300000$ (such as when training large models like GPT3). The Moments Accountant was subsumed into the framework of Renyi Differential Privacy (RDP) introduced by~\cite{mironov2017renyi}. The running time of these accountants are independent of $k$, but they only give an upper bound and cannot approximate the privacy curve to arbitrary accuracy.

 DP-SGD is one of the most important DP algorithms in practice, because one can use it to train neural networks to achieve good privacy-vs-utility tradeoffs. Therefore obtaining accurate and tight privacy guarantees for DP-SGD is important. For example reducing $\eps$ from $2$ to $1$, can mean that one can train the network for 4 times more epochs while staying within the same privacy budget. Therefore DP-SGD is one of the main motivations for this work. 

 There are also situations when the PRVs do not have bounded moments and so Moments Accountant or Renyi DP cannot be applied for analyzing privacy. An example of such an algorithm is the DP-SGD-JL algorithm of \cite{bu2021fast} which uses numerical composition of PRVs to analyze privacy.

\paragraph{GDP Accountant} \cite{DongRS19,BuDLS19} introduced the notion of Gaussian Differential Privacy (GDP) and used it to develop an accountant for DP-SGD. The accountant is based on central limit theorem and only gives an approximation to the true privacy curve, where the approximation gets better with $k$. But as we show in Figure~\ref{fig:ours_vs_Koskela}, GDP accountant can significantly underreport the true epsilon value.

Several different notions of privacy were introduced for obtaining good upper bounds on the privacy curve of composition of DP algorithms such as Concentrated DP (CDP)~\cite{DR16,BS16}, Truncated CDP~\cite{bun2018composable} etc. None of these methods can approximate the privacy curve of compositions to arbitrary accuracy. The notion of $f$-DP introduced by~\cite{DongRS19}, allows for a lossless composition theorem, but computing the privacy curve of composition seems computationally hard and they do not give any algorithms for doing it.

\subsection{Our Contributions}
The main contribution of this work is a new algorithm with an improved analysis for computing the privacy curve of the composition of a large number of DP algorithms.
\begin{theorem}[Informal version of Theorem~\ref{thm:approximation_eps_upper_bound}]
    \label{thm:informal_main}
	Suppose $M_1,M_2,\dots,M_k$ are DP algorithms. Then the privacy curve $\delta_M(\eps)$ of adaptive composition $M=M_1 \circ M_2 \circ \dots \circ M_k$ can be approximated in time
	\begin{equation}
		\label{eqn:time_ours}
		O\lp\frac{\eps_\upper \ k^{1.5}\log{k}\sqrt{\log \frac{1}{\delta_\error}}}{\eps_\error}\rp,
	\end{equation}
	where $\eps_\error$ is the additive error in $\eps$, $\delta_\error$ is the additive error in $\delta$ and $\eps_\upper$ is an upper bound on $\max\left\{\eps_{M}(\delta_\error),\max_i \eps_{M_i}\lp\frac{\delta_\error}{k}\rp\right\}.$\footnote{$\eps_{M}(\delta)$ is the inverse of $\delta_{M}(\eps)$.}
\end{theorem}
If each $M_i$ satisfies $\lp\frac{1}{\sqrt{k}},\frac{o(1)}{k}\rp$-DP, then by advanced composition (Proposition~\ref{prop:advanced_composition}), we can set $\eps_\upper=O(1)$. Therefore the running time of our algorithm in this case is $\tilde{O}\lp \frac{k^{1.5}\sqrt{\log \frac{1}{\delta_\error}}}{\eps_\error}\rp.$
We can save a factor of $k$, when we compose the same algorithm with itself $k$ times.
\begin{theorem}
	Suppose $M$ is a DP algorithm. Then the privacy curve $\delta_{M^{\circ k}}(\eps)$ of $M$ (adaptively) composed with itself $k$ times can be approximated in time
	\begin{equation}
		\label{eqn:time_ours_single_alg}
		O\lp\frac{\eps_\upper \ k^{\frac{1}{2}}\log{k}\sqrt{\log \frac{1}{\delta_\error}}}{\eps_\error}\rp,
	\end{equation}
	where $\eps_\error$ is the additive error in $\eps$, $\delta_\error$ is the additive error in $\delta$ and $\eps_\upper$ is an upper bound on $\max\left\{\eps_{M^{\circ k}}(\delta_\error), \eps_M\lp\frac{\delta_\error}{k}\rp\right\}.$
\end{theorem}

Thus we improve the state-of-the-art by at least a factor of $k$ in running time. We also note that our algorithm improves the memory required by a factor of $k.$ See Figure~\ref{fig:ours_vs_Koskela} for a comparison of our algorithm with that of~\cite{KoskelaJPH21}. Also note that RDP Accountant (equivalent to the Moments Accountant) significantly overestimates the true $\eps$, while the GDP Accountant significantly underestimates the true $\eps.$ In contrast, the upper and lower bounds provided by our algorithm lie very close to each other.

\begin{figure}[!h]
    %\centering
    \begin{subfigure}[b]{0.45\textwidth}
        \resizebox{\linewidth}{!}{\large \input{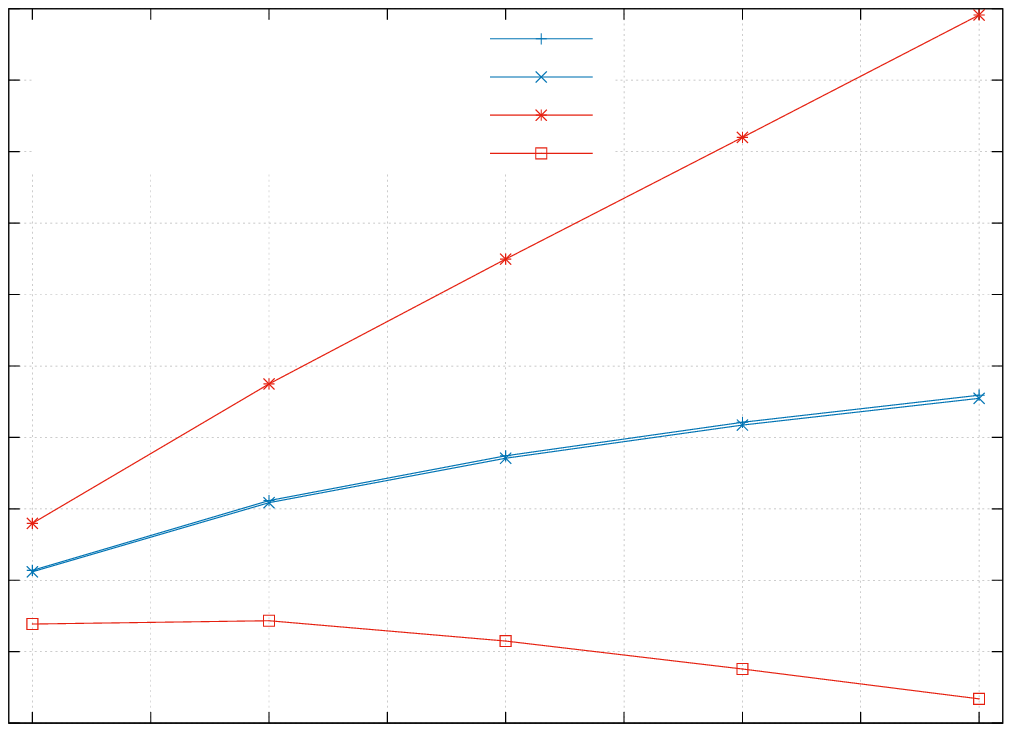}}
        %\input{Figures/dpsgd-vs_fourier}
    % \resizebox{!}{.3\paperwidth}{\large \input{Figures/dpsgd-vs_fourier}}
 \caption{Our algorithm gives much closer upper and lower bounds on the true privacy curve compared to~\cite{KoskelaJPH21}, under the same mesh size of $4\times 10^{-5}$. Our upper and lower bounds are nearly coinciding.}
    %\caption{Setting $p=10^{-3}$, $\sigma=0.8$, $\delta=10^{-7}$, mesh size of $4 \times 10^{-5}$. $\eps_{\mathrm{up}}$ and $\eps_{\mathrm{low}}$ are upper and lower bounds on the true privacy curve $\eps(\cdot).$}
    \end{subfigure}
    \qquad
    \begin{subfigure}[b]{0.45\textwidth}
        \resizebox{\linewidth}{!}{\large  \input{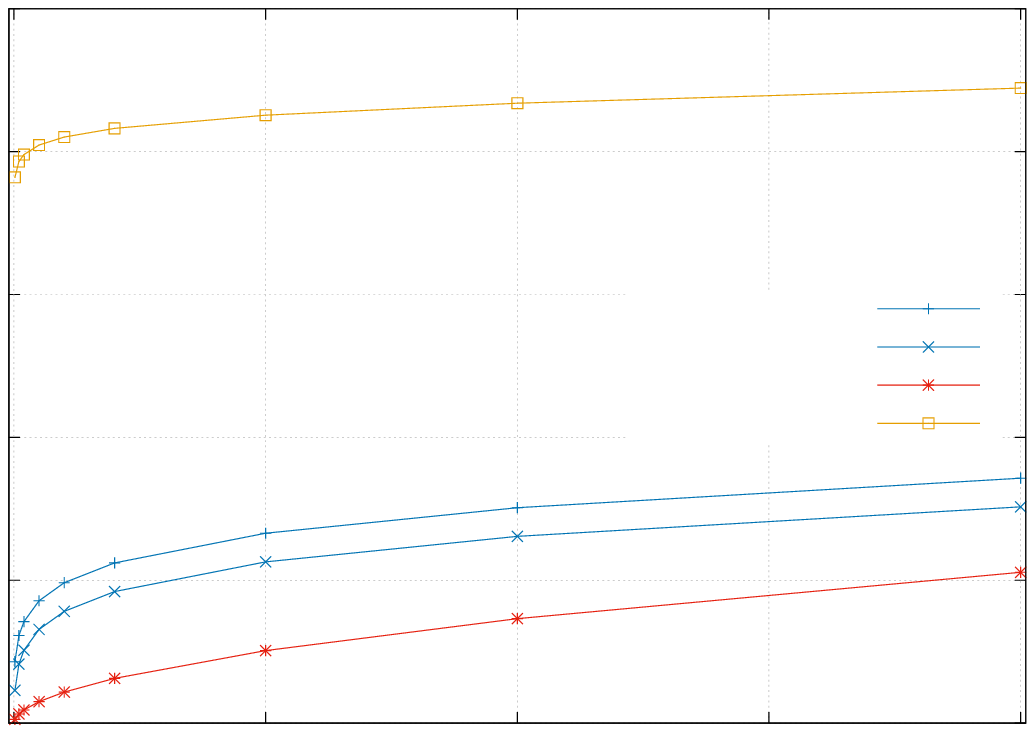}}
         %\input{Figures/dpsgd_example}
    %\resizebox{!}{.3\paperwidth}{\large \input{Figures/dpsgd_example}}
   
    %\caption{Case study: DPSGD.
    %A realistic example of a sampling probability of \num{1e-3}, noise multiplier of $0.8$ and $\delta$ of \num{1e-7}.
    \caption{Our algorithm can improve significantly over the RDP Accountant. We also see that GDP Accountant can significantly underreport the true $\eps$. We have set $\eps_\error=0.1$, $\delta_\error=\delta/1000$ here.}
    \end{subfigure}
    \caption{Case study on DP-SGD. Sampling probability $p=10^{-3}$, noise scale $\sigma=0.8$, $\delta=10^{-7}$.}
    \label{fig:ours_vs_Koskela}
\end{figure}

\paragraph{Our Techniques} Our algorithm (also the prior work of~\cite{koskela2020computing}) proceeds by approximating the privacy loss random variables (PRVs) by truncating and discretizing them. We then use Fast Fourier Transform (FFT) to convolve the distributions efficiently. The main difference is in the approximation procedure and the error analysis. In the approximation procedure, we correct the approximation so that the expected value of the discretization matches with the expected value of the PRV.

 To analyze the approximation error, we introduce the concept of \emph{coupling approximation} (Definition~\ref{def:couplingapx}), which is a variant of Wasserstein (optimal transport) distance specifically tailored to this application.
We first show that the approximation output by our algorithm to each privacy random variable is a good coupling approximation. We then show that when independent coupling approximations are added, cancellation happens between the errors due to Hoeffding bound, producing a much better coupling approximation than one naively expects from the triangle inequality. This allows us to choose the mesh size in our discretization to be $\approx \frac{1}{\sqrt{k}}$, whereas~\cite{koskela2021computing} choose a mesh size of $\approx \frac{1}{k}$. The other improvement is the truncation procedure. We give a tight tail bound of the PRVs (Lemma~\ref{lem:PRVs_tailbounds}). This allows us to choose the domain size for in truncation to be $\approx \tilde{O}(1)$, whereas~\cite{koskela2021computing} choose $\approx \tilde{O}(\sqrt{k})$. Both ideas together saves a factor of $k$ in the run time and memory. 

For the analysis, the previous paper analyzes the discretization error by studying the stability of convolution. This leads to complicated calculations with the runtime linear in $1/{\delta_\error}$ (see (\ref{eqn:time_Koskela})). Since $\delta_\error \ll \delta \ll 1/N$ is required to give meaningful privacy guarantee ($N$ is the number of users), this term $1/{\delta_\error}$ is huge. In this paper, we show various facts about how coupling approximation accumulates and use them to give a runtime depending only on $\sqrt{\log(1/{\delta_\error})}$.

\section{DP Preliminaries}
\label{sec:prelims}
%!TEX root=./MAIN.tex
%The $f$-DP framework, recently proposed in~\cite{DongRS19}, allows us to reason about a collection of $(\eps,\delta)$-privacy guarantees simultaneously which can then be composed to get a much better $(\eps,\delta)$-privacy for the final algorithm. 
%We will first define the notion of $(\eps,\delta)$-DP formally and then define the notion of $f$-DP. We then state a proposition from~\cite{DongRS19} which shows that these two notions are dual to each other.

%We will begin with definition of $(\eps,\delta)$-differentially private algorithm.
%\begin{definition}[$(\eps,\delta)$-DP~\cite{DMNS06}]
%	We say that an algorithm $\cM$ is $(\eps,\delta)$-DP if for any two neighboring databases $D,D'$ and any subset $S$ of outputs, we have $\Pr[\cM(D)\in S] \le e^\eps\Pr[\cM(D')\in S]+\delta.$
%\end{definition}

Given a DP algorithm $\cM$, for each value of $\eps\ge 0$, there exists some $\delta \in [0,1]$ such that $\cM$ is $(\eps,\delta)$-DP. We can represent all these privacy guarantees by a function $\delta_\cM(\eps):\R^{\ge 0} \to [0,1]$ and say that $\delta_\cM(\cdot)$ is the \emph{privacy curve} of $\cM.$ This inspires the following definition of a privacy curve between two random variables.

\begin{definition}[Privacy curve]
	Given two random variables $X,Y$ supported on some set $\Omega$, define $\delta(X||Y): \R \to [0,1]$ as:
	$$\delta(X||Y)(\eps) = \sup_{S\subset \Omega} \Pr[Y\in S] - e^\eps \Pr[X\in S].$$
\end{definition}

Therefore an algorithm $\cM$ is $(\eps,\delta)$-DP iff $\delta\lp \cM(D)||\cM(D')\rp(\eps) \le \delta$ for all neighboring databases $D,D'.$ 

\begin{remark}
	Note that not all functions $\delta:\R \to [0,1]$ are privacy curves. A characterization of privacy curves can be obtained using the $f$-DP framework of~\cite{DongRS19}. The notion of privacy curve $\delta(X||Y)$ and tradeoff function $T(X||Y)$ are dual to each other via convex duality~\cite{DongRS19}. This implies a characterization of privacy curves as shown in \cite{zhu2021optimal}.
\end{remark}

\begin{definition}[Composition of privacy curves~\cite{DongRS19}]
	Let $\delta_1\equiv \delta(X_1||Y_1)$ and $\delta_2\equiv \delta(X_2||Y_2)$ be any two privacy curves. The composition of the privacy curves, denoted by $\delta_1 \otimes \delta_2$, is defined as $$\delta_1 \otimes \delta_2 \equiv \delta\lp(X_1,X_2)||(Y_1,Y_2)\rp$$ where $X_1,X_2$ are independently sampled and $Y_1,Y_2$ are independently sampled.
\end{definition}
Note that there can be many pairs of random variables which have the same privacy curve, but the above operation is well-defined. If $\delta(X_1||Y_1)\equiv \delta(X_1'||Y_1')$ and $\delta(X_2||Y_2)\equiv \delta(X_2'||Y_2')$, then it was shown by~\cite{DongRS19} that $$\delta\lp(X_1,X_2)||(Y_1,Y_2)\rp=\delta\lp(X_1',X_2')||(Y_1',Y_2')\rp.$$ \cite{DongRS19} also show that $\otimes$ is a commutative and associative operation.

Given two DP algorithms $M_1$ and $M_2$, the adaptive composition $(M_2 \circ M_1)(D)$ is an algorithm which outputs $(M_1(D),M_2(D,M_1(D))$, i.e., $M_2$ can look at the database $D$ and also the output of the previous algorithm $M_1(D)$. Adaptive composition of more than two algorithms is similarly defined. Suppose $M_1$ has privacy curve $\delta_1$ and $M_2$ has privacy curve $\delta_2$ (i.e., $M_2(\cdot, y)$ is a DP algorithm with privacy curve $\delta_2$ for any fixed $y.$). The following composition theorem shows how to get the privacy curve of $M_2 \circ M_1.$
\begin{theorem}[Composition theorem~\cite{DongRS19}]
\label{prop:composition}
	Let $M_1,M_2,\dots,M_k$ be DP algorithms with privacy curves given by $\delta_1,\delta_2,\dots,\delta_k$ respectively. The privacy curve of the adaptive composition $M_k\circ M_{k-1} \circ \dots \circ M_1$ is given by $\delta_1\otimes \delta_2 \otimes \cdots \otimes \delta_k.$
\end{theorem}

%\newpage
\section{Privacy Loss Random Variables (PRVs)}
\label{sec:PRV}
%!TEX root=./MAIN.tex

%Given a privacy curve $\delta$, there can be many pairs of random variables $(X,Y)$ such that $\delta(X||Y)\equiv \delta$. 
The notion of \emph{privacy loss random variables} (PRVs) is a unique way to assign a pair $(X,Y)$ for any privacy curve $\delta$ such that $\delta \equiv \delta(X||Y).$ PRVs allow us to compute composition of two algorithms via summing random variables (Theorem~\ref{thm:composition_PRVs}) (equivalently, convolving their distributions). Thus PRVs can be thought of as a \emph{reparametrization of privacy curves} where composition becomes convolution. In this paper, we differ from the usual definition of PRVs given in~\cite{DR16,koskela2020computing}, which are tied to a specific algorithm. Instead we think of them as a reparametrization of privacy curves and study them directly. This allows us to succinctly prove many useful properties of PRVs.

Let $\bR=\R\cup\set{-\infty,\infty}$ be the extended real line where we define $\infty+x=\infty$ and $-\infty+x=-\infty$ for $x\in \R.$

\begin{definition}[Privacy loss random variables (PRVs)]
	Given a privacy curve $\delta:\R\to[0,1]$, we say that $(X,Y)$ are \emph{privacy loss random variables} for $
	\delta$, if they satisfy the following conditions:
	\begin{itemize}
	 	\item $X,Y$ are supported on $\bR$,
	 	\item $\delta(X||Y)\equiv \delta$,
	 	\item $Y(t)=e^t X(t)$ for every $t\in \R$ and
	 	\item $Y(-\infty)=0$ and $X(\infty)=0$
	 \end{itemize} 
	 where $X(t),Y(t)$ are probability density functions of $X,Y$ respectively.
\end{definition}
Mathematically, the correct way to write the condition $Y(t)=e^t X(t)$ is to say that $\E_{Y}[\phi(Y)] = \E_{X}[\phi(X)e^X]$ for all test functions $\phi:\bR \to [0,1]$ with $\phi(\infty)=\phi(-\infty)=0.$ This will generalize to all situations where $X,Y$ are continuous or discrete or both. 
%Equivalently, one can rewrite the condition as $\frac{dY}{dX}(t)=e^t$ where $\frac{dY}{dX}$ is the Radon-Nikodym derivative of the distribution of $Y$ w.r.t to the distribution of $X$. 
For ease of exposition, we ignore this complication and assume that $X(t),Y(t)$ represent the PDFs if $X,Y$ are continuous at $t$, or the probability masses if they have point masses at $t$.

The following theorem shows that the PRVs for a privacy curve $\delta=\delta(P||Q)$ are given by the log-likelihood random variables of $P,Q.$ 
\begin{restatable}{theorem}{PRVLLRV}
\label{thm:PRV_LLRV}
 Let $\delta:\R \to [0,1]$ be a privacy curve given by $\delta \equiv \delta(P||Q)$ where $P,Q$ are two random variables supported on $\Omega$. The PRVs $(X,Y)$ for the privacy curve $\delta$ are given by\footnote{Here $Q(\omega)$ and $P(\omega)$ are the probability density functions of $Q,P$ respectively. Note that the mathematically precise way is to replace the ratio $\frac{Q(\omega)}{P(\omega)}$ by the Radon-Nikodym derivative $\frac{dQ}{dP}(\omega).$}:
 $$X = \log\lp \frac{Q(\omega)}{P(\omega)}\rp \text{ where } \omega \sim P,$$
 $$Y = \log\lp \frac{Q(\omega)}{P(\omega)}\rp \text{ where } \omega \sim Q.$$
 % Moreover, the privacy curve $\delta$ can be expressed in terms of PRVs $(X,Y)$ as:
 % \begin{equation}
 % \label{eqn:delta_PRV_simple}
 % \delta(\eps)=\Pr[Y> \eps] - e^\eps \Pr[X > \eps]=\E_Y[(1-e^{\eps-Y})_+] =\Pr[Y\ge \eps+Z].
 % \end{equation}
 % where $Z$ is an exponential random variable.\footnote{For $x\in \R, x_+=\max\{x,0\}$.}
\end{restatable}

The following theorem provides a formula for computing the privacy curve $\delta$ in terms of the PRVs and conversely a formula for PRVs in terms of the privacy curve. A similar statement appears in~\cite{sommer2019privacy,koskela2020computing}.

\begin{restatable}{theorem}{PRVprivacycurve}
\label{thm:PRV_privacycurve}
 The privacy curve $\delta$ can be expressed in terms of PRVs $(X,Y)$ as:
 \begin{equation}
 \label{eqn:delta_PRV_simple}
 \delta(\eps)=\Pr[Y> \eps] - e^\eps \Pr[X > \eps]=\E_Y[(1-e^{\eps-Y})_+] =\Pr[Y\ge \eps+Z].
 \end{equation}
 where $Z$ is an exponential random variable.\footnote{For $x\in \R, x_+=\max\{x,0\}$.} Conversely, given a privacy curve $\delta:\R\to [0,1]$, we can compute the PDFs of its PRVs $(X,Y)$ as:
 \begin{equation}
 \label{eqn:PRV_from_delta}
 Y(t) = \delta''(t)-\delta'(t) \text{ and } X(t)=e^t(\delta''(t)-\delta'(t)).
 \end{equation}
\end{restatable}

\begin{remark}
	Theorem~\ref{thm:PRV_privacycurve} shows that the PRVs $X,Y$ do not depend on the particular $P,Q$ used to represent the privacy curve $\delta$ in Theorem~\ref{thm:PRV_LLRV}. So we should think of the PDF of of the PRV $Y$ (or $X$) as an equivalent reparametrization of the privacy curve $\delta:\R \to [0,1]$, just as the notion of $f$-DP~\cite{DongRS19} is a reparametrization of the privacy curve $\delta$. %In particular any random variable $Y$ supported on $\R\cup {\infty}$ which satisfies $\E[e^{-Y}]=1$ is the PRV of some privacy curve.\footnote{This condition arises because $X(t)=e^tY(t)$ and $\E[e^{-Y}]=\int_{-\infty}^{\infty} e^{-t}Y(t) dt =\int_{-\infty}^{\infty} X(t) dt =1$.}
\end{remark}

PRVs are useful in computing privacy curves because the composition of two privacy curves can be computed by adding the corresponding pairs of PRVs. A similar statement appears in~\cite{DR16}.
\begin{restatable}{theorem}{compositePRV}
\label{thm:composition_PRVs}
Let $\delta_1,\delta_2$ be two privacy curves with PRVs $(X_1,Y_1)$ and $(X_2,Y_2)$ respectively. Then the PRVs for $\delta_1 \otimes \delta_2=\delta(X_1,X_2||Y_1,Y_2)$ are given by $(X_1+X_2,Y_1+Y_2)$. In particular, $$\delta_1\otimes \delta_2 = \delta(X_1+X_2||Y_1+Y_2).$$
\end{restatable}
\begin{proof}
Let $(X,Y)$ be the privacy random variables for $\delta(X_1,X_2||Y_1,Y_2)$. By Theorem~\ref{thm:PRV_LLRV},
\begin{align*}
	X & = \log\lp\frac{(Y_1,Y_2)(t_1,t_2)}{(X_1,X_2)(t_1,t_2)}\rp \text{ where } (t_1,t_2)\sim (X_1,X_2)\\
	& = \log\lp\frac{Y_1(t_1)Y_2(t_2)}{X_1(t_1)X_2(t_2)}\rp \text{ where } t_1\sim X_1,t_2\sim X_2\tag{By independence of $X_1,X_2$ and indpendence of $Y_1,Y_2$}\\
	& = \log\lp e^{t_1} \cdot e^{t_2}\rp \text{ where } t_1\sim X_1,t_2\sim X_2\\
	& = t_1+t_2 \text{ where } t_1\sim X_1,t_2\sim X_2\\
	& = X_1 + X_2.
\end{align*}
Similarly, \begin{align*}
	Y & = \log\lp\frac{(Y_1,Y_2)(t_1,t_2)}{(X_1,X_2)(t_1,t_2)}\rp \text{ where } (t_1,t_2)\sim (Y_1,Y_2)\\
	& = t_1+t_2 \text{ where } t_1\sim Y_1,t_2\sim Y_2\\
	& = Y_1 + Y_2.
\end{align*}
\end{proof}

In Appendix~\ref{sec:appendix_PRV}, we provide a proof of Theorems \ref{thm:PRV_LLRV} and \ref{thm:PRV_privacycurve}. We also discuss how to compute the PRVs for a subsampled mechanism given the PRVs for the original mechanism and give examples of PRVs for few standard mechanisms. These are used in our experiments to calculate the PRVs for DP-SGD.

\section{Numerical composition of privacy curves}
\label{sec:numerical_composition}
%!TEX root=./MAIN.tex

In this section, we present an efficient and numerically accurate method, \textsf{ComposePRV} (Algorithm~\ref{alg:PrivacyComposition}), for composing privacy guarantees by utilizing the notion of PRVs.

\begin{algorithm}[h]

\SetAlgoLined
\KwIn{$\CDF$s of PRVs $Y_1,Y_2,\dots,Y_k$, mesh size $h$, Truncation parameter $L\in \frac{h}{2}+h\Z^{> 0}$}
\KwOut{PDF of an approximation $\tY$ for $Y=\sum_{i=1}^k Y_i$. $\tY$ will be supported on $\mu+(h\Z \cap [-L,L])$ for some $\mu\in [0,\frac{h}{2}]$. %And an approximation to the privacy curve $\delta_{Y}(\eps)$.\gnote{We don't need an approximation if we are given $L$. The approximation is only needed to set $L.$}
}
%\KwResult{}
\For{$\ell=1$ to $k$}{
	$\tY_i \leftarrow \mathsf{DiscretizePRV}(Y_i,L,h)$\;
}
 
Compute PDF Of $\tY = \tY_1 \oplus_L \tY_2 \oplus_L \cdots \oplus_L \tY_k$ by convolving PDFs of $\tY_1,\tY_2,\dots,\tY_k$ using FFT\;
Compute $\delta_{\tY}(\eps) = \E_{\tY}\lb \lp 1-e^{\eps-\tY}\rp_+\rb$ for all $\eps \in [0,L]$\;
Return $\tY,\delta_{\tY}(\cdot)$
 \caption{\textsf{ComposePRV}: Composing privacy curves using PRVs}
 \label{alg:PrivacyComposition}
\end{algorithm}

In the algorithm \textsf{ComposePRV}, we compute the circular convolution $\oplus_L$ using Fast Fourier Transform (FFT). Fix some $L>0$. For $x\in \R$, we define $x \pmod{2L}=x-2Ln$ where $n\in \Z$ is chosen such that $x-2Ln \in (-L,L].$ Given $x,y$, we define the circular addition $$x\oplus_L y = x+y \pmod{2L}.$$ When we use FFT to compute the convolution of two discrete distributions $Y_1,Y_2$ supported on $h\Z \cap [-L,L]$, we are implicitly calculating the the distribution of $Y_1 \oplus_L Y_2$. 
In the appendix, we show that $\tY_1 \oplus_L \tY_2 \oplus_L \dots \oplus_L \tY_k$ is a good approximation of $Y_1+Y_2+\dots+Y_k$.

The subroutine \textsf{DiscretizePRV} (Algorithm~\ref{alg:discretizePRV}) is used to truncate and discretize PRVs. In this subroutine, we shift the discretized random variables such that it has the same mean as the original variables. This is one of main differences between our algorithm and the algorithm in \cite{KoskelaJPH21, koskela2021computing}. We show that this significantly decreases the discretization error and allow us to use much coarser mesh $h \approx 1/\sqrt{k}$ instead of $h \approx 1/k$.

\begin{algorithm}[h]
\SetAlgoLined
\KwIn{$\CDF_Y(\cdot)$ of a PRV $Y$, mesh size $h$, Truncation parameter $L\in \frac{h}{2}+h\Z^{> 0}$}
\KwOut{PDF of an approximation $\tY$ supported on $\mu+(h\Z \cap [-L,L])$ for some $\mu \in [0,\frac{h}{2}].$}
%\KwResult{}

 $n \leftarrow \frac{L- \frac{h}{2}}{h}$\;
 
 \For{$i=-n$ to $n$}{
	 $q_i \leftarrow \CDF_Y(ih+h/2)-\CDF_Y(ih-h/2)$\;
 }

 $q \leftarrow q/\lp \sum_{i=-n}^n q_i\rp$ \tcp*{Normalize $q$ to make it a probability distribution}

 % //Calculate mean of $Y$\;
 $Y^L \leftarrow Y\big|_{|Y|\le L}$ (i.e., $Y$ conditioned on $|Y|\le L$)\;
 %$E\lb Y\ \big|\ |Y|\le L\rb \leftarrow$ mean of $Y$ conditioned on $|Y|\le L$\;
 %L(\CDF_Y(L)+\CDF_Y(-L))-\int_{-L}^L \CDF_Y(t) dt$\; 
 $\mu \leftarrow \E[Y^L]- \sum_{i=-n}^{n} ih \cdot q_i$\;
 $\tY \leftarrow \begin{cases}
 	ih+\mu &\text{ w.p. } q_i \text{ for } -n\le i\le n
 \end{cases}$\;
 Return $\tY$\;
 \caption{\textsf{DiscretizePRV}: Discretize and truncate a PRV}
 \label{alg:discretizePRV}
\end{algorithm}

For simplicity, throughout this paper, we will assume that the PRVs $Y_1,Y_2,\dots,Y_k$ do not have any mass at $\infty$. This is with out loss of generality. Suppose $\Pr[Y_i=\infty]=\delta_i$ for each $i.$ Let $Y_i'=Y_i|_{Y_i\ne \infty}$. Then 
\begin{align*}
	Y_1+Y_2+\dots+Y_k =
	\begin{cases}
		Y_1'+Y_2'+\dots+Y_k' & \text{ w.p. } (1-\delta_1)(1-\delta_2)\cdots(1-\delta_k)\\
		\infty & \text{ w.p. } 1-(1-\delta_1)(1-\delta_2)\cdots(1-\delta_k).
	\end{cases}
\end{align*}
Therefore we can use Algorithm~\ref{alg:PrivacyComposition} to approximate the distribution of $Y_1'+Y_2'+\dots+Y_k'$, and use it to approximate the distribution of $Y_1+Y_2+\dots+Y_k$.

\section{Error analysis}
\label{sec:error_analysis}
To analyze the discretization error, we introduce the notion of \emph{coupling approximation}, a variant of Wasserstein distance. Intuitively, a good coupling approximation is a coupling where the two random variables are close to each other with high probability.
\begin{definition}[coupling approximation]\label{def:couplingapx}
	Given two random variables $Y_1,Y_2$, we write $|Y_1-Y_2|\le_{\eta} h$ if there exists a coupling between $Y_1,Y_2$ such that $\Pr[|Y_1-Y_2|>h]\le \eta.$
\end{definition}
The following lemma shows that if we have a good coupling approximation $\tY$ to a PRV $Y$, then the privacy curves $\delta_Y(\eps)$ and $\delta_{\tY}(\eps)$ should be close. 
\begin{lemma}
	\label{lem:coupling_to_privacy_curves}
	If $Y$ and $\tY$ are two random variables such that $|Y-\tY|\le_\eta h$, then for every $\eps\in \R$, 
	$$\delta_{\tY}(\eps+h)-\eta \le \delta_{Y}(\eps) \le \delta_{\tY}(\eps-h)+\eta.$$
	% where $\delta_Y(\eps)=\Pr[Y\ge \eps+Z]$ for an independent exponential random variable $Z$ and $\delta_{\tY}$ is defined similarly.
\end{lemma}
\begin{proof} By Theorem \ref{thm:PRV_LLRV}, $\delta_Y(\eps)=\Pr[Y\ge \eps+Z]$ and hence
\begin{align*}
	\delta_Y(\eps) % &=\Pr[Y\ge \eps+Z] \\
	&=\Pr[Y-\tY + \tY \ge \eps +Z]\\
	&\le\Pr[Y-\tY \ge h] + \Pr[\tY \ge \eps -h +Z]\\
	&\le\eta+ \delta_{\tY}(\eps-h).
\end{align*}
Similarly, we have $\delta_{\tY}(\eps) \le \eta + \delta_{Y}(\eps-h)$ for all $\eps\in \R.$ 
\end{proof}

Therefore the goal of our analysis is to show that the \textsf{ComposePRV} algorithm finds a good coupling approximation $\tY$ to $Y=\sum_{i=1}^k Y_i.$ We first show that the \textsf{DiscretizePRV} algorithm computes a good coupling approximation to the PRVs and crucially, it preserves the expected value after truncation. Lemma~\ref{lem:discretizePRV_coupling} shows that $|\tY-Y^L|\le_0 h$ where $\tY$ is the approximation of a PRV $Y$ output by Algorithm~\ref{alg:discretizePRV} and $Y^L$ is the truncation of $Y$ to $[-L,L].$

We then use the following key lemma which shows that when we add independent coupling approximations (where expected values match), we get a much better coupling approximation than what the triangle inequality predicts.
\begin{lemma}
	\label{lem:coupling_sum_independent}
	Suppose $Y_1,Y_2,\dots,Y_k$ and $\tY_1,\tY_2,\dots,\tY_k$ are two collections of independent random variables such that $|Y_i-\tY_i|\le_0 h$ and $\E[Y_i]=\E[\tY_i]$ for all $i$, then $$\left|\sum_{i=1}^k Y_i - \sum_{i=1}^k \tY_i\right|\le_{\eta}h\sqrt{2k\log{\frac{2}{\eta}}}.$$
\end{lemma}
\begin{proof}
	Let $X_i=Y_i-\tY_i$ where $(Y_i,\tY_i)$ are coupled such that $|Y_i-\tY_i|\le h$ w.p. $1$. Then $X_i \in [-h,h]$ w.p. $1$. Note that $X_1,X_2,\dots,X_k$ are independent of each other. By Hoeffding's inequality, 
	$$\Pr\lb\left|\sum_i X_i\right| \ge t\rb \le 2\exp\lp- \frac{2t^2}{k(2h)^2}\rp=\eta$$ if we set $t=h\sqrt{2k\log{\frac{2}{\eta}}}$.
\end{proof}

This lemma shows that the error of $k$ times composition is around $\sqrt{k}\cdot h$ and hence setting $h \approx 1/\sqrt{k}$ gives small enough error. Next, we bound the domain size $L$. Naively, the domain size $L$ should be of the order of $\sqrt{k}$ because $Y$ is the sum of $k$ independent random variables with each bounded by a constant. In the appendix, we give a tighter tail bound of $Y$.
\begin{restatable}{lemma}{PRVtail}
	\label{lem:PRVs_tailbounds}
	Let $(X,Y)$ be the privacy random variables for a $(\epsilon, \delta)$-DP algorithm, then for any $t\ge 0$, we have
	$$\Pr[|Y|\ge \eps + t ] \le \frac{\delta \lp 1+e^{-\eps-t}\rp}{1-e^{-t}}.$$
	%\text{ and }$$
	%$$\Pr[|X|\ge \eps + t ] \le \frac{\delta \lp 1+e^{-\eps-t}\rp}{1-e^{-t}}.$$
\end{restatable}

This shows that $\Pr[|Y| \ge \eps + 2] \le \frac{4}{3} \delta$ and hence truncating the domain with $L = 2 + \eps$ only introduces an additive $\delta$ error in the privacy curve. Therefore, if the composition satisfies a good privacy guarantee (namely $\eps = O(1)$ for small enough $\delta$), we can truncate the domain at $L = \Theta(1)$. Together with the fact that mesh size is $1/\sqrt{k}$, this gives a $O(\sqrt{k})$-time algorithm for computing the privacy curve when we compose the same mechanism with itself $k$ times. The following theorem gives a formal statement of the error bounds of our algorithm, it is proved in Appendix~\ref{sec:error_analysis}.

\begin{restatable}{theorem}{mainthm}
	\label{thm:approximation_eps_upper_bound}
	Let $\eps_\error,\delta_\error>0$ be some fixed error terms.
	 Let $\cM_1,\cM_2,\dots,\cM_k$ be DP algorithms with privacy curves $\delta_{\cM_i}(\eps)$. Let $Y_i$ be the PRV corresponding to $\cM_i$ such that $\delta_{\cM_i}(\eps)=\delta_{Y_i}(\eps)$ for $\eps\ge 0$. Let $\cM$ be the (adaptive) composition of $\cM_1,\cM_2,\dots,\cM_k$ and let $\delta_{\cM}(\eps)$ be its privacy curve. 
	 Set $L\ge 2+\eps_\error$ sufficiently large such that 
	\begin{equation}
	\label{eqn:L_delta}
		 \sum_{i=1}^k \delta_{\cM_i}(L-2) \le \frac{\delta_\error}{8} \text{ and } \delta_{\cM}(L-2-\eps_\error) \le \frac{\delta_\error}{4}.
	\end{equation}
	Let $\tY$ be the approximation of $Y=\sum_{i=1}^k Y_i$ produced by \textsf{ComposePRV} algorithm with mesh size $$h=\frac{\eps_\error}{\sqrt{\frac{k}{2}\log \frac{12}{\delta_\error}}}.$$ Then 
		\begin{equation}
	\label{eqn:upperlower_delta}
	\delta_{\tY}(\eps+\eps_{\error}) -\delta_{\error}\le \delta_Y(\eps) = \delta_{\cM}(\eps) \le \delta_{\tY}(\eps-\eps_{\error})+\delta_\error.
	\end{equation}
Furthermore, our algorithm takes $O\lp b \frac{L}{h} \log\lp\frac{L}{h}\rp\rp$ time where $b$ is the number of distinct algorithms among $\cM_1,\cM_2,\dots,\cM_k$.	
\end{restatable}

\begin{remark}
	A simple way to set $L$ such that the condition (\ref{eqn:L_delta}) holds is by choosing an $L$ such that:
\begin{equation}
\label{eqn:L_eps}
 L\ge 2+\max\left\{\eps_\error+\eps_{\cM}\lp \frac{\delta_\error}{4}\rp, \max_{i\in [k]}\ \eps_{\cM_i}\lp \frac{\delta_\error}{8k}\rp\right\}
\end{equation}
where $\eps_{\cA}(\delta)$ is the inverse of $\delta_\cA(\eps)$. To set the value of $L$, we do not need the exact value of $\eps_{\cM}$ (or $\eps_{\cM_i}$). We only need an upper bound on $\eps_{\cM}$, which can often by obtained by using the RDP Accountant or any other method to derive upper bounds on privacy.
\end{remark}

\section{Experiments}
\label{sec:experiments}
%!TEX root=./MAIN.tex

In this section, we demonstrate the utility of our composition method by computing the privacy curves for the DP-SGD algorithm which is one of the most important algorithms in differential privacy.

The DP-SGD algorithm~\cite{Abadi16} is a variant of stochastic gradient descent with $k$ steps. In each step, the algorithm selects a $p$ fraction of training examples uniformly at random. The algorithm adds a Gaussian vector with variance $\propto \sigma^2$ to the clipped gradient of the selected batch. Then it performs a gradient step (or any other iterative methods) using the noisy gradient computed. The privacy loss of DP-SGD involves composing the privacy curve of each iteration with itself $k$ times. The PRVs for each iteration have a closed form and depend only $p,\sigma$ (see Appendix). Our algorithms use this closed form of PRVs.

See Figure~\ref{fig:ours_vs_Koskela}(b) for the comparison between our algorithm and the GDP and RDP Accountant.
Our method provides a lower and upper bound of the privacy curve according to \eqref{eqn:upperlower_delta}.  
%The figure shows that our method significantly improves upon the upper bounds on privacy loss provided by the Moments Accountant of~\cite{Abadi16}. The figure also shows that GDP Accountant can underestimate the true privacy loss by a significant amount. 
In Figure~\ref{fig:ours_vs_Koskela}(a), we compare our algorithm with~\cite{KoskelaJPH21} (implemented in~\cite{PLDAccountant}).
Under the same mesh size, our algorithm computes a much closer upper and lower bound. 

We validate our program for the case $p=1$. When $p=1$, we have an exact formula for 

\begin{equation}
\label{eqn:analytical_solution_p=1}
\delta(\eps) = \Phi\lp -\frac{\eps}{\mu}+\frac{\mu}{2} \rp - e^\eps \Phi\lp -\frac{\eps}{\mu}-\frac{\mu}{2} \rp    
\end{equation}
 where $\mu=\frac{\sqrt{k}}{\sigma}$. In  Figure~\ref{fig:validation}, we show that the true privacy curve is indeed sandwiched between the bounds we compute and that the vertical distance between our bounds is indeed $2 \epsilon_\error$ with a neglible $\delta_\error$ of $10^{-10}$.

\begin{figure}[h]
    \centering
    \resizebox{!}{0.4\textwidth}{\large \input{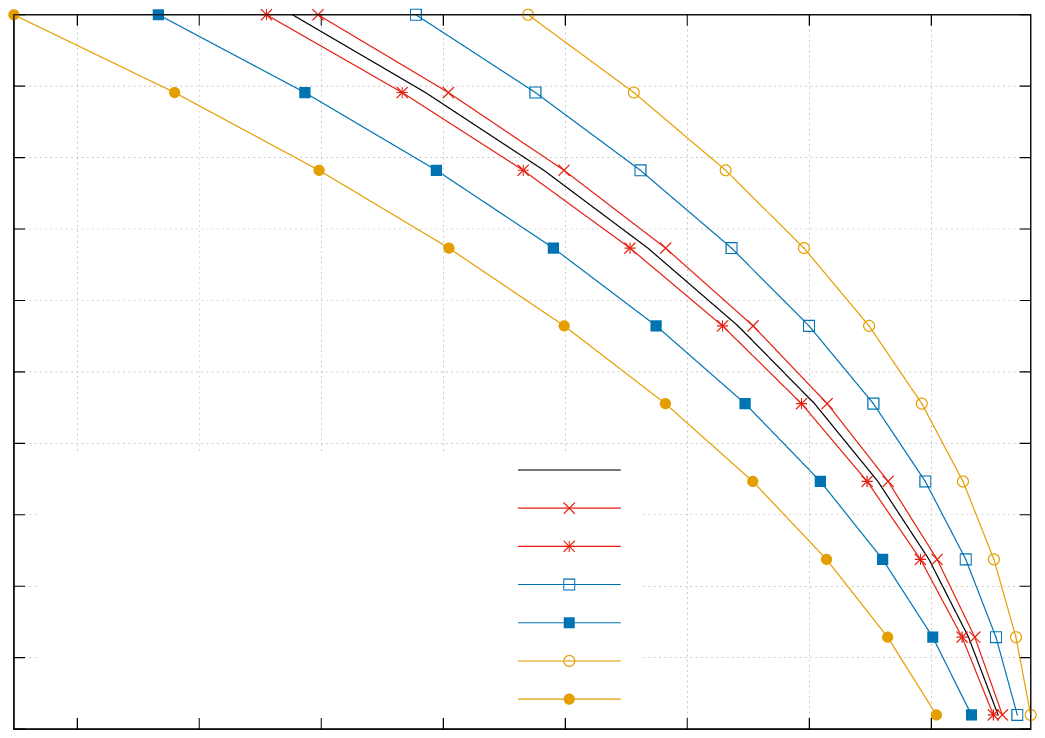}}
    \caption{Setting $p=1$ and comparing to the analytical solution~(\ref{eqn:analytical_solution_p=1}).} 
    \label{fig:validation}
\end{figure}

\paragraph{Floating point errors} Note that our error analysis in Section~\ref{sec:error_analysis} ignores floating point errors. This is because they are negligible compared to the discretization and truncation errors we analyzed in Section~\ref{sec:error_analysis} for the range of $\delta$ we are interested in. Our implementation uses {\verb long } {\verb double } floating point format which is platform dependent, however, it guarantees a precision at least as good as double precision which has a resolution of $10^{-15}$.
Computations involving $\delta$ of these orders of magnitude suffer from floating point inaccuracies.
Our implementation therefore only allows $\delta$ values which are greater than $10^{-10}$ which sufficies for practical use cases. See Appendix~\ref{sec:numerical_precision} for more details.

\subsection{Comparison with \cite{KoskelaJPH21}}
\label{sec:runtime_experiments}

\begin{figure}[h]
    \centering
    \resizebox{!}{0.4\textwidth}{\large \input{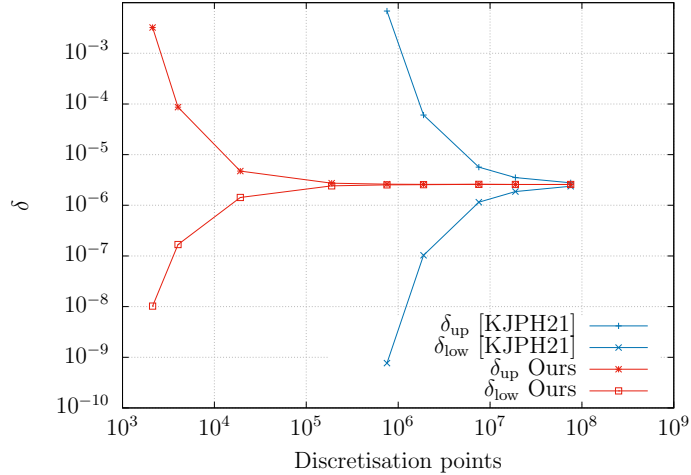}}
    \caption{Comparison of error bounds of $\delta$ with varying number of discretisation points for $p=4\times10^{-3}, \sigma=0.8, \varepsilon=1.5, k=1000$.} 
    \label{fig:comparen}
\end{figure}

\begin{figure}[h]
    \begin{subfigure}[b]{0.33\textwidth}
        \hspace{-1.5cm}\resizebox{1.7\linewidth}{!}{\input{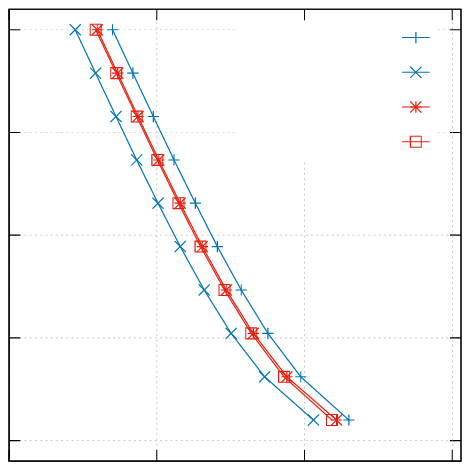}}
        \caption{$k=10$}
    \end{subfigure}
    \begin{subfigure}[b]{0.33\textwidth}
        \hspace{-1.5cm}\resizebox{1.7\linewidth}{!}{\input{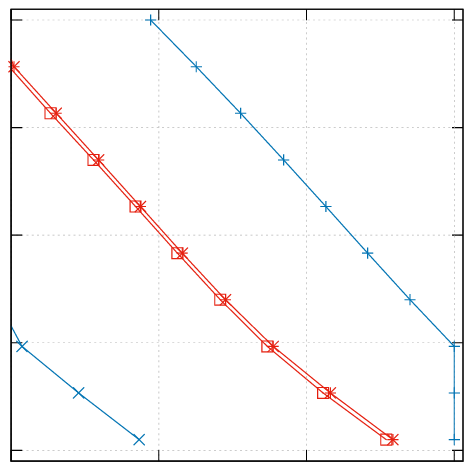}}
        \caption{$k=100$}
    \end{subfigure}
    \begin{subfigure}[b]{0.33\textwidth}
        \hspace{-1.5cm}\resizebox{1.7\linewidth}{!}{\input{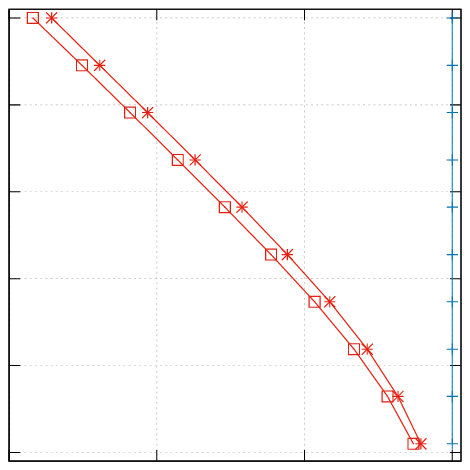}}
        \caption{$k=1000$}
    \end{subfigure}
    \caption{
        Comparing different error bounds using the same mesh size $8\times10^{-4}$ under different number of steps $k = 10, 100, 1000$. (With $p=10^{-2}$, $\sigma=0.8$.)
        \label{fig:comparek}}
\end{figure}

\begin{figure}[h]
    %\centering
    \begin{subfigure}[b]{0.49\textwidth}
        \resizebox{\linewidth}{!}{\large  \input{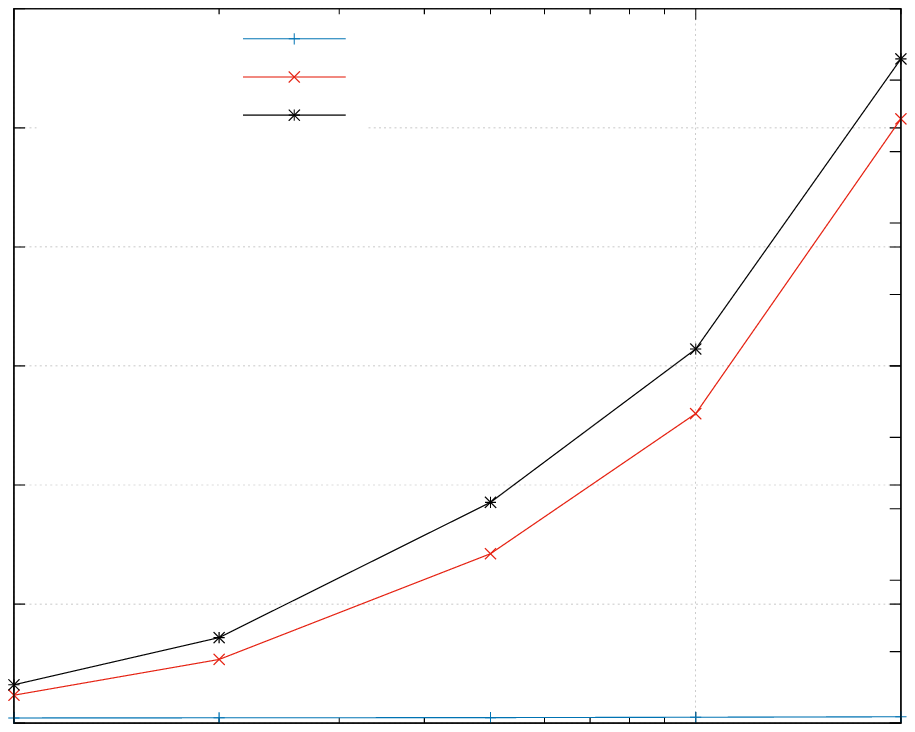}}

    \end{subfigure}
        \qquad
    \begin{subfigure}[b]{0.49\textwidth}
        \resizebox{\linewidth}{!}{\large \input{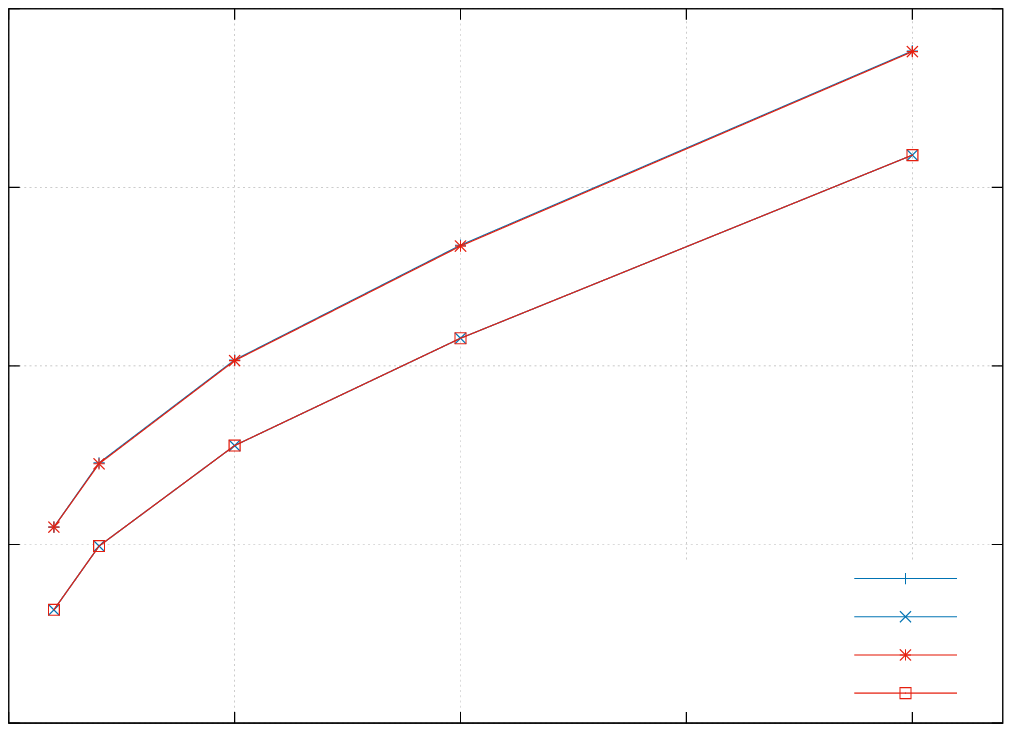}}

    \end{subfigure}

    \caption{(a) Comparing runtimes for our algorithm with that of~\cite{KoskelaJPH21} when aligned on accuracy for $\sigma = 0.8$, $p=4\times10^{-3}$.
We can see a significant reduction in runtime in particular for large number of DPSGD steps. We were not able to run the algorithm of~\cite{KoskelaJPH21} beyond 2,000 steps, since it becomes unstable beyond that point.\protect\footnotemark.
We also plot the speed up directly on the secondary $y$-axis. (b) Verification of the alignment of the error bounds of both algorithms at $\varepsilon=1.5$.}
    \label{fig:runtimes}
\end{figure}
\footnotetext{We are using the implementation of~\cite{KoskelaJPH21} from~\cite{PLDAccountant}.}

In this section, we provide more results demonstrating the practical use of our algorithm.
We compare runtimes of our algorithm with~\cite{KoskelaJPH21}, which is the state-of-the-art, for typical values of privacy parameters ($\sigma = 0.8$, $p=4\times10^{-3}$, $\varepsilon=1.5$).

See Figure~\ref{fig:comparen} for the effect of the number of discretisation points $n$ on the accuracy of $\delta$.
Our algorithm requires about a few orders of magnitude smaller number of discretization points to converge compared to the algorithm of \cite{KoskelaJPH21}.
A similar picture can be seen in Figure~\ref{fig:comparek}.
While for a small number of compositions, the algorithm of ~\cite{KoskelaJPH21} gives reasonable estimates, for a large number of compositions, their error bounds worsen quickly.

We note that runtimes are directly proportional to the memory required by the algorithms and so a separate memory analysis is not required; the runtime and memory are dominated by the number of points in the discretization of PRV.
All experiments are performed on a Intel Xeon W-2155 CPU with 3.30GHz with 128GB of memory.

In order to compare runtimes, we align the accuracy of both FFT algorithms.
We find sets of numerical parameters (number of discretization bins and domain length) such that both algorithms give similarly accurate bounds and verify it visually (see Figure \ref{fig:runtimes} (b)).
Figure \ref{fig:runtimes} illustrates the runtimes for varying numbers of DPSGD steps.
We observe a significant reduction in the runtime using our algorithms.

\subsection*{Acknowledgements}
We would like to thank Janardhan Kulkarni and Sergey Yekhanin for several useful discussions and encouraging us to work on this problem.
L.W. would like to thank Daniel Jones and Victor R\"uhle for fruitful discussions and helpful guidance.

\FloatBarrier
\bibliographystyle{alpha}
\bibliography{references}

\appendix

%!TEX root=./MAIN.tex

\section{Effect of floating point arithmetic}
\label{sec:numerical_precision}

In this section, we demonstrate the effect of floating point inaccuracies on the computed privacy parameters.
Figure \ref{fig:num_error} compares lower and upper bounds of the privacy curve with the analytical solution for small values of $\delta$.
As mentioned in section \ref{sec:experiments}, we use a floating point representation with a resolution of at least $10^{-15}$.
The number of discretization points in this examples are on the order of $10^4$.
Consequently, we expect floating point inaccuracies to become dominant for values on the order of $10^{-11}$.
This can be also seen in the illustration, where the lower and upper bound fail to produce meaningful results for $\delta < 2 \times 10^{-11}$.

\begin{figure}[h]
    \centering
    \resizebox{!}{0.4\textwidth}{\large \input{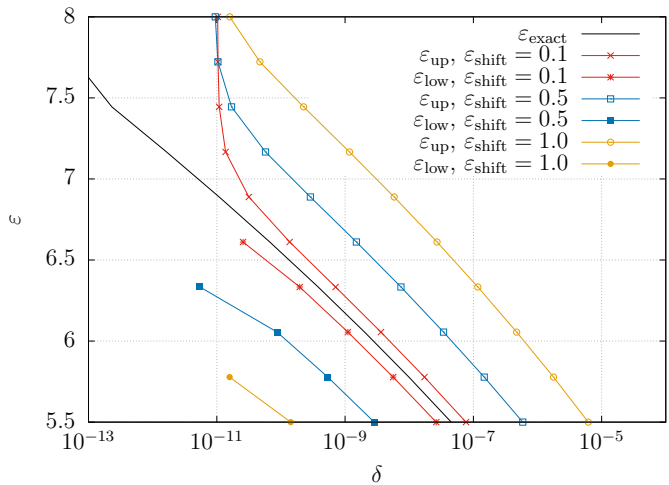}}
    \caption{Setting $p=1$ and comparing to the analytical solution~(\ref{eqn:analytical_solution_p=1}) for values of $\delta$ beyond expected floating point accuracy.} 
    \label{fig:num_error}
\end{figure}

%!TEX root=./MAIN.tex
\section{Privacy Loss Random Variables}
\label{sec:appendix_PRV}
In this section, we continue the discussion on privacy random variables in Section \ref{sec:PRV}. First, we give the proof of the formula for PRVs of $\delta(P||Q)$ and the formula for a privacy curve given its PRVs (Theorem \ref{thm:PRV_LLRV}).

{\renewcommand\footnote[1]{}\PRVLLRV*}
\begin{proof}
We will first verify that $Y(t)=e^t X(t).$ This is equivalent to proving that $\E_Y[\phi(Y)]=\E_X[\phi(X)e^X]$ for any test function $\phi:\bR \to [0,1]$. This is true since
	\begin{align*}
		\E_Y[\phi(Y)]&= \E_{\omega\sim Q}\lb \phi\lp\log\lp \frac{Q(\omega)}{P(\omega)}\rp\rp\rb \\
		&= \E_{\omega\sim P}\lb \phi\lp\log\lp \frac{Q(\omega)}{P(\omega)}\rp\rp \frac{Q(\omega)}{P(\omega)}\rb\\
		&= \E_X[\phi(X)e^X].
	\end{align*}
	We will now prove that $\delta(X||Y)=\delta(P||Q).$
	We have
	\begin{align*}
		\delta(P||Q)(\eps) &= \sup_{S \subset \Omega} \Pr[Q\in S] - e^\eps\Pr[P\in S]\\
		&= \Pr[Q\in S_\eps] - e^\eps\Pr[P\in S_\eps]
	\end{align*}
	where $$S_\eps = \left\{\omega\in \Omega: \frac{Q(\omega)}{P(\omega)} > e^\eps \right\}=\left\{\omega\in \Omega: \log\lp\frac{Q(\omega)}{P(\omega)}\rp > \eps \right\}.$$ 
	Therefore $\Pr[Q\in S_\eps]=\Pr[Y > \eps]$ and $\Pr[P\in S_\eps]=\Pr[X > \eps]$. To complete the proof, note that 
	\begin{align*}
		\delta(X||Y)(\eps) &= \sup_{T \subset \bR} \Pr[Y\in T] - e^\eps\Pr[X\in T]\\
		&= \Pr[Y\in T_\eps] - e^\eps\Pr[X\in T_\eps]
	\end{align*}
	where $$T_\eps = \left\{t\in \bR: \frac{Y(t)}{X(t)} > e^\eps \right\}= \left\{t\in \bR: e^t > e^\eps \right\}=(\eps,\infty].$$
	Putting it all together, we have:
	\begin{align*}
		\delta(P||Q)(\eps)&=\Pr[Y>\eps] - e^\eps \Pr[X>\eps]=\delta(X||Y).
	\end{align*}
% Finally, since the PDFs of PRVs $(X,Y)$ satisfy the relation $Y(t)=e^tX(t)$, we can rewrite the equation~\ref{eqn:delta_PRV_simple} in terms of just $Y$ or just $X$.
% 		\begin{align*}
% 		\delta(\eps)&=\Pr[Y\ge \eps]-e^\eps \Pr[X\ge \eps]\\
% 		&=\int_\eps^\infty Y(t) dt - \int_\eps^\infty e^\eps X(t) dt\\
% 		&=\int_\eps^\infty Y(t) dt - \int_\eps^\infty e^{\eps-t} Y(t) dt \tag{Since $Y(t)=e^tX(t)$}\\
% 		&=\int_\eps^\infty Y(t) (1-e^{\eps-t}) dt\\
% 		&=\int_{-\infty}^\infty Y(t) (1-e^{\eps-t})_+ dt\\
% 		&=\E_Y[(1-e^{\eps-Y})_+]
% 	\end{align*}
% 	To get the other form for $\delta(\eps)$, we use the integration by parts formula.
% 	\begin{align*}
% 		\delta(\eps)&=\int_{\eps}^\infty Y(t) (1-e^{\eps-t}) dt\\
% 		&=\int_{\eps}^\infty Y(t) dt + \int_{\eps}^\infty \lp -Y(t) \rp e^{\eps-t} dt\\
% 		&=\Pr[Y\ge \eps] + \lp \Pr[Y\ge t]e^{\eps-t}\Big|_\eps^\infty - \int_{\eps}^\infty \Pr[Y\ge t] \lp -e^{\eps-t}\rp dt \rp\\
% 		&=\Pr[Y\ge \eps] - \Pr[Y\ge \eps] + \int_{\eps}^\infty \Pr[Y\ge t] e^{\eps-t} dt\\
% 		&=\int_\eps^\infty e^{\eps-t} \Pr[Y\ge t] dt \\
% 		&=\int_0^\infty e^{-z} \Pr[Y\ge \eps+z] dz  \tag{Substituting $z=t-\eps$}\\
% 		&=\Pr[Y\ge \eps+Z]. \tag{where $Z$ is an exponential random variable}
% 	\end{align*}
\end{proof}

{\renewcommand\footnote[1]{}\PRVprivacycurve*}
\begin{proof}	
Since the PDFs of PRVs $(X,Y)$ satisfy the relation $Y(t)=e^tX(t)$, we can rewrite the equation~\ref{eqn:delta_PRV_simple} in terms of just $Y$ or just $X$.
		\begin{align*}
		\delta(\eps)&=\Pr[Y\ge \eps]-e^\eps \Pr[X\ge \eps]\\
		&=\int_\eps^\infty Y(t) dt - \int_\eps^\infty e^\eps X(t) dt\\
		&=\int_\eps^\infty Y(t) dt - \int_\eps^\infty e^{\eps-t} Y(t) dt \tag{Since $Y(t)=e^tX(t)$}\\
		&=\int_\eps^\infty Y(t) (1-e^{\eps-t}) dt\\
		&=\int_{-\infty}^\infty Y(t) (1-e^{\eps-t})_+ dt\\
		&=\E_Y[(1-e^{\eps-Y})_+]
	\end{align*}
	To get the other form for $\delta(\eps)$, we use the integration by parts formula.
	\begin{align*}
		\delta(\eps)&=\int_{\eps}^\infty Y(t) (1-e^{\eps-t}) dt\\
		&=\int_{\eps}^\infty Y(t) dt + \int_{\eps}^\infty \lp -Y(t) \rp e^{\eps-t} dt\\
		&=\Pr[Y\ge \eps] + \lp \Pr[Y\ge t]e^{\eps-t}\Big|_\eps^\infty - \int_{\eps}^\infty \Pr[Y\ge t] \lp -e^{\eps-t}\rp dt \rp\\
		&=\Pr[Y\ge \eps] - \Pr[Y\ge \eps] + \int_{\eps}^\infty \Pr[Y\ge t] e^{\eps-t} dt\\
		&=\int_\eps^\infty e^{\eps-t} \Pr[Y\ge t] dt \\
		&=\int_0^\infty e^{-z} \Pr[Y\ge \eps+z] dz  \tag{Substituting $z=t-\eps$}\\
		&=\Pr[Y\ge \eps+Z]. \tag{where $Z$ is an exponential random variable}
	\end{align*}
We now prove the converse relation by differentiating the expression for $\delta(\eps)$ twice. We have:
	\begin{align*}
 	&\delta(\eps)=\int_\eps^\infty Y(t) dt - e^{\eps}\int_\eps^\infty e^{-t} Y(t) dt\\
	\implies & \delta'(\eps)=-Y(\eps)+e^{\eps}\cdot e^{-\eps}Y(\eps)-e^{\eps}\cdot \int_\eps^\infty e^{-t} Y(t) dt = -e^{\eps}\cdot \int_\eps^\infty e^{-t} Y(t) dt \\
	\implies & e^{-\eps}\delta'(\eps)=-\int_\eps^\infty e^{-t} Y(t) dt \\
	\implies & e^{-\eps}\delta''(\eps)-e^{-\eps}\delta'(\eps)=e^{-\eps} Y(\eps) \\
	\implies & Y(\eps)=\delta''(\eps)-\delta'(\eps).
	\end{align*}

\end{proof}

\subsection{Examples of privacy loss random variables}
In this section, we state the PRVs for a few standard mechanisms. 

\begin{proposition}[Gaussian Mechanism]
The PRVs for $\delta(\cN(\mu,1)||\cN(0,1))$ are:
$$X=\cN(-\mu^2/2,\mu^2) \text{ and } Y=\cN(\mu^2/2,\mu^2).$$
\end{proposition}
\begin{proof}
	Let $P=\cN(\mu,1)$ and $Q=\cN(0,1)$.
	By Theorem~\ref{thm:PRV_LLRV},
	\begin{align*}
		Y &\sim \log\lp \frac{Q(t)}{P(t)}\rp \text{ where } t\sim Q\\
		&\sim \log\lp \frac{\exp(-t^2/2)}{\exp(-(t-\mu)^2/2)}\rp \text{ where } t\sim \cN(0,1)\\
		&\sim \frac{(t-\mu)^2}{2} - \frac{t^2}{2} \text{ where } t\sim \cN(0,1)\\
		&\sim \frac{\mu^2}{2}-\mu t \text{ where } t\sim \cN(0,1)\\
		&= \cN\lp \frac{\mu^2}{2},\mu^2 \rp.
	\end{align*}
	A similar calculation shows that $X=\cN\lp -\frac{\mu^2}{2},\mu^2\rp$
\end{proof}

\begin{proposition}[Laplace Mechanism]
	The PRVs for the privacy curve $\delta\lp\Lap\lp \mu,1\rp||\Lap\lp 0,1\rp\rp$ are:
	$$X=|Z|-|Z-\mu| \text{ and } Y= |Z-\mu|-|Z|$$ where $Z\sim \Lap(0,1).$
\end{proposition}
\begin{proof}
	Let $P=\Lap(\mu,1)$ and $Q=\Lap(0,1)$.
	By Theorem~\ref{thm:PRV_LLRV},
	\begin{align*}
		Y &\sim \log\lp \frac{Q(t)}{P(t)}\rp \text{ where } t\sim Q\\
		&\sim \log\lp \frac{\exp(-|t|)}{\exp(-|t-\mu|}\rp \text{ where } t\sim \Lap(0,1)\\
		&\sim |t-\mu|-|t| \text{ where } t\sim \Lap(0,1)\\
		&= |Z-\mu|-|Z| \text{ where } Z\sim \Lap(0,1).
	\end{align*}
	A similar calculation shows that $X=|Z|-|Z-\mu|$ where $Z\sim\Lap(0,1).$
\end{proof}

\begin{proposition}[$(\eps,\delta)$-DP]
	The PRVs for the privacy curve of a $(\eps,\delta)$-DP algorithm are
	$$X=\begin{cases}
	-\infty & \text{w.p. } \delta\\
	-\eps & \text{w.p. }  \frac{(1-\delta)e^\eps}{e^\eps+1}\\
	\eps & \text{w.p. } \frac{1-\delta}{e^\eps+1},
	\end{cases}$$
	$$Y=\begin{cases}
	-\eps & \text{w.p. } \frac{1-\delta}{e^\eps+1}\\
	\eps & \text{w.p. }  \frac{(1-\delta)e^\eps}{e^\eps+1}\\
	\infty & \text{w.p. } \delta.
	\end{cases}$$
\end{proposition}
\begin{proof}
	It is easy to verify that $Y(t)=e^t X(t)$ for all $t\in \R$. We can also verify that $$\delta(\eps)=\Pr[Y> \eps] - e^\eps \Pr[X > \eps]=\delta.$$ Morever $X=-Y$, therefore the privacy curve $\delta(X||Y)$ is symmetric by Proposition~\ref{prop:finv_prv}, i.e., $\delta(X||Y)=\delta(Y||X)$. These conditions together imply that $X,Y$ are PRVs for the $(\eps,\delta)$-DP curve.
\end{proof}

Note that in the all the above examples, we have $X=-Y$ as the privacy curves are symmetric.

\subsection{Subsampling}
In this section, we calculate the PRVs for a subsampled mechanism given the PRVs for the original mechanism. Given two random variables $P,Q$ and a sampling probability $p\in [0,1]$,  $p\cdot P+(1-p)\cdot Q$ denotes the mixture where we sample $P$ w.p. $p$ and $Q$ w.p. $1-p.$
\begin{proposition}
	\label{prop:mixture_f_p_privacy_rv}
	Let $(X,Y)$ be the PRVs for a privacy curve $\delta(P||Q)$. Let $(X_p,Y_p)$ be the PRVs for $\delta_p=\delta(P||\ p\cdot P+(1-p)\cdot Q)$. Then
	\begin{align*}
		X_p&=\log(1+p(e^X-1)),\\
		Y_p&=\begin{cases}
			\log(1+p(e^Y-1)) \text{ w.p. } p\\
			\log(1+p(e^X-1)) \text{ w.p. } 1-p.
		\end{cases}
	\end{align*}
	The CDFs of $X_p$ and $Y_p$ are given by:
	\begin{align*}
		\CDF_{X_p}(t) &= \begin{cases}
			\CDF_X\lp\log\lp \frac{e^t-(1-p)}{p}\rp \rp &\text{ if } t\ge \log(1-p)\\
			0 &\text{ if } t< \log(1-p)
		\end{cases}\\
		\CDF_{Y_p}(t) &= \begin{cases}
			p\cdot \CDF_Y\lp \log\lp \frac{e^t-(1-p)}{p}\rp \rp + (1-p)\cdot \CDF_X\lp \log\lp \frac{e^t-(1-p)}{p}\rp \rp &\text{ if } t\ge \log(1-p)\\
			0 & \text{ if } t <\log(1-p).
		\end{cases}
	\end{align*}

\end{proposition}

\begin{proof}
By Theorem~\ref{thm:PRV_LLRV}, 
\begin{align*}
	X_p &= \log\lp \frac{pY(t)+(1-p)X(t)}{X(t)}\rp \text{ where } t\sim X\\
	& = \log \lp p e^t + 1-p\rp \text{ where } t \sim X\\
	& = \log \lp p e^X + 1-p\rp.
\end{align*}
Similarly,
\begin{align*}
	Y_p &= \log\lp \frac{pY(t)+(1-p)X(t)}{X(t)}\rp \text{ where } t\sim pY+(1-p)X\\
	& = \log \lp p e^t + 1-p\rp \text{ where } t \sim pY+(1-p)X\\
	& = \begin{cases}
			\log(1+p(e^Y-1)) \text{ w.p. } p\\
			\log(1+p(e^X-1)) \text{ w.p. } 1-p.
		\end{cases} 
\end{align*}
The CDF of $X_p$ is given by:
\begin{align*}
	\Pr[X_p \le t] &= \Pr\lb \log \lp p e^X + 1-p\rp \le t\rb\\
	&= \Pr\lb X \le \log\lp \frac{e^t-(1-p)}{p}\rp \rb
\end{align*}
The CDF of $Y_p$ is given by:
\begin{align*}
	\Pr[Y_p \le t] &= p \Pr\lb \log \lp p e^Y + 1-p\rp \le t\rb + (1-p) \Pr\lb \log \lp p e^X + 1-p\rp \le t\rb\\
	&= p\Pr\lb Y \le \log\lp \frac{e^t-(1-p)}{p}\rp \rb + (1-p)\Pr\lb X \le \log\lp \frac{e^t-(1-p)}{p}\rp \rb.
\end{align*}
\end{proof}

%!TEX root=./MAIN.tex

\section{Missing Proofs in Error Analysis}
\label{sec:error_analysis_extra}

\subsection{Facts about Coupling Approximation}
Here we collect some useful properties of coupling approximations.
The following lemma shows that the coupling approximations satisfy a triangle inequality.
\begin{lemma}[Triangle inequality for couplings]
	\label{lem:coupling_triangle_ineq}
	Suppose $X,Y,Z$ are random variables such that $|X-Y|\le_{\eta_1} h_1$ and $|Y-Z|\le_{\eta_2} h_2$. Then $|X-Z|\le_{\eta_1 +\eta_2} h_1+h_2.$
\end{lemma}
\begin{proof}
	There exists couplings $(X,Y)$ and $(Y,Z)$ such that
	$$\Pr[|X-Y|\ge h_1] \le \eta_1 \text{ and } \Pr[|Y-Z|\ge h_2] \le \eta_2.$$ From these two couplings, we can construct a coupling between $(X,Z)$: sample $X$, sample $Y$ from $Y|X$ (given by coupling $(X,Y)$) and finally sample $Z$ from $Z|Y$ (given by coupling $(Y,Z)$). Therefore for this coupling, we have:
	\begin{align*}
	 \Pr[|X-Z|\ge h_1+h_2] &\le \Pr[|(X-Y) + (Y-Z)|\ge h_1+h_2]\\
	  &\le \Pr[|X-Y| + |Y-Z|\ge h_1+h_2]\\
	  &\le \Pr[|X-Y|\ge h_1]+\Pr[|Y-Z|\ge h_2]\\
	  &\le \eta_1+\eta_2.
	\end{align*}
\end{proof}

The following lemma shows that small total variation distance implies good coupling approximation.
\begin{lemma}
	\label{lem:TV_coupling}
	If the total variation distance $d_{TV}(X,Y)\le \eta$, then $|X-Y| \le_\eta 0$.
\end{lemma}
\begin{proof}
	It is well known that for any two random variables $X,Y$, there exists a coupling such that $d_{TV}(X,Y)=\Pr[X\ne Y]$. This immediately implies what we want.
\end{proof}

\subsection{Bounding the error using tail bounds of PRVs}
The goal of this section is to bound the error of \textsf{ComposePRV} in terms of the tail bounds of the underlying PRVs.
\begin{theorem}
	\label{thm:approximation}
	Let $Y_1,Y_2,\dots,Y_k$ be PRVs and let $\tY$ be the approximation produced by the \textsf{ComposePRV} algorithm (Algorithm~\ref{alg:PrivacyComposition}) for $Y=\sum_{i=1}^k Y_i$ with truncation parameter $L$ and mesh size $$h=\frac{\eps_\error}{\sqrt{\frac{k}{2}\log \frac{2}{\eta_0}}}.$$ Then $$\delta_{\tY}(\eps+\eps_{\error}) -\delta_{\error}\le \delta_Y(\eps) \le \delta_{\tY}(\eps-\eps_{\error})+\delta_\error$$ where 
	% $$\delta_\error= \inf_{\lambda >0} \frac{\prod_{i=1}^k\E[\exp(\lambda Y_i)]}{e^{\lambda L}} + \inf_{\lambda >0} \frac{\prod_{i=1}^k\E[\exp(-\lambda Y_i)]}{e^{\lambda L}} + \sum_{i=1}^k \Pr[|Y_i|\ge L]+ \eta_0.$$
	\begin{align*}
		\delta_\error&= \Pr\lb\left|\sum_{i=1}^k \tY_i\right|\ge L\rb + \sum_{i=1}^k \Pr[|Y_i|\ge L]+ \eta_0\\
		&\le \Pr\lb\left|\sum_{i=1}^k Y_i\right|\ge L-\eps_\error\rb + 2\sum_{i=1}^k \Pr[|Y_i|\ge L]+ 2\eta_0.
	\end{align*}
\end{theorem}

\begin{remark}
	We can directly bound $\Pr\lb\left|\sum_{i=1}^k \tY_i\right|\ge L\rb$ using moment generating functions as $$\Pr\lb\left|\sum_{i=1}^k \tY_i\right|\ge L\rb \le \inf_{\lambda >0} \frac{\prod_{i=1}^k\E[\exp(\lambda \tY_i)]}{e^{\lambda L}} + \inf_{\lambda >0} \frac{\prod_{i=1}^k\E[\exp(-\lambda \tY_i)]}{e^{\lambda L}}.$$ Sometimes, if we already have good upper bound for $\Pr\lb \left|\sum_i Y_i \right|\ge L\rb$, then the second bound on $\delta_\error$ is useful.
\end{remark}

The following key lemma shows that the \textsf{DiscretizePRV} algorithm (Algorithm~\ref{alg:discretizePRV}) produces a good coupling approximation to the PRV and preserves the mean. 
\begin{lemma}
	\label{lem:discretizePRV_coupling}
	Given a PRV $Y$, let $Y^L = Y\big|_{|Y|\le L}$ be its truncation. The approximation $\tY$ returned by \textsf{DiscretizePRV} satisfies $\E[\tY]=\E[Y^L]$ and $|Y^L-(\tY-\mu)|\le_{0} \frac{h}{2}$ for some $\mu$ where $h$ is the mesh size. We also have $|Y^L - Y|\le_\eta 0$ where $\eta=\Pr[|Y|\ge L].$
\end{lemma}
\begin{proof}
Since $d_{TV}(Y,Y^L) \le \Pr[|Y|\ge L]=\eta$, by Lemma~\ref{lem:TV_coupling}, $|Y-Y^L|\le_\eta 0$. It is clear that by construction $\E[\tY]=\E[Y^L],$
$$\E[\tY]=\mu + \sum_{i=-n}^{n} ih \cdot q_i = \lp \E[Y^L] - \sum_{i=-n}^{n} ih \cdot q_i\rp + \sum_{i=-n}^{n} ih \cdot q_i =\E[Y^L].$$ We will now construct the coupling between $(Y^L,\tY)$ such that $|Y^L-(\tY-\mu)|\le \frac{h}{2}$. The coupling is as follows: First sample $y\sim Y^L$. Suppose $y\in (ih-\frac{h}{2},ih+\frac{h}{2}]$ for some integer $i$ such that $-n\le i \le n$, then return $\ty = \mu + ih$. Clearly, the distribution of $\ty$ matches with $\tY$ and $|y-(\ty-\mu)| = |y-ih|\le \frac{h}{2}$.

\end{proof}

Since our error bound on $\tY$ is slightly different from the assumption in Lemma \ref{lem:coupling_sum_independent}, we need the following generalization using the same proof.

\begin{lemma}
	\label{lem:coupling_sum_independent2}
	Suppose $Y_1,Y_2,\dots,Y_k$ and $\tY_1,\tY_2,\dots,\tY_k$ are two collections of independent random variables such that $|Y_i-(\tY_i - \mu_i)|\le_0 h$ for some $\mu_i$ and $\E[Y_i]=\E[\tY_i]$ for all $i$, then $$\left|\sum_{i=1}^k Y_i - \sum_{i=1}^k \tY_i\right|\le_{\eta}h\sqrt{2k\log{\frac{2}{\eta}}}.$$
\end{lemma}

In the algorithm, we only calculate the distribution of $Y_1 \oplus Y_2 \oplus \dots \oplus Y_k$ instead of $Y_1+Y_2+\dots+Y_k$. The following simple lemma shows that this is still a good approximation as long as $\sum_i Y_i$ stays within $[-L,L]$ with high probability.

\begin{lemma}
\label{lem:convolution_wrapping_around_error}
Let $Y_1,Y_2,\dots,Y_k$ be random variables supported on $(-L,L]$. 
Then $$\left|\sum_{i=1}^k Y_i - \lp Y_1 \oplus_L Y_2 \oplus_L \cdots \oplus_L Y_k\rp\right| \le_\eta 0$$ where $$\eta = \Pr\lb\left|\sum_{i=1}^k Y_i\right|\ge L\rb.$$
\end{lemma}
\begin{proof}
$$\Pr\lb\sum_{i=1}^k Y_i \ne \lp Y_1 \oplus_L Y_2 \oplus_L \cdots \oplus_L Y_k\rp \rb \le \Pr\lb\left|\sum_{i=1}^k Y_i\right|\ge L\rb.$$ This clearly implies what we want.
\end{proof}

Combining all the above lemmas, we get the following corollary.
\begin{corollary}
\label{cor:composePRV_coupling}
Let $Y_1,Y_2,\dots,Y_k$ be random variables supported on and let $\tY_i$ be the discretization of $Y_i$ produced by \textsf{DiscretizePRV} algorithm with mesh size $h=\frac{h_0}{\sqrt{\frac{k}{2}\log \frac{2}{\eta_0}}}$ and truncation parameter $L$. Then
$$\left|(Y_1+Y_2+\dots+Y_k) - (\tY_1 \oplus \tY_2 \oplus \dots \oplus \tY_k)\right|\le_\eta h_0$$ where 
 $$\eta = \Pr\lb\left|\sum_{i=1}^k \tY_i\right|\ge L\rb + \sum_{i=1}^k \Pr[|Y_i|\ge L]+ \eta_0.$$ Furthermore, we can bound $$\Pr\lb\left|\sum_{i=1}^k \tY_i\right|\ge L\rb \le \Pr\lb\left|\sum_{i=1}^k Y_i\right|\ge L-h_0\rb+\sum_{i=1}^k \Pr[|Y_i|\ge L]+ \eta_0.$$
\end{corollary}
\begin{proof}
	Let $Y^L\equiv Y_i\big|_{|Y_i|\le L}$ be the truncation of $Y_i.$
	By Lemma~\ref{lem:discretizePRV_coupling}, $|Y_i^L - (\tY_i -\mu_i)| \le_{0} \frac{h}{2}$ for some $\mu_i$ and $|Y_i^L - Y_i|_{\xi_i} \le 0$ where $\xi_i = \Pr[|Y_i|\ge L].$ Now applying the triangle inequality for coupling approximations (Lemma~\ref{lem:coupling_triangle_ineq}), we have 
	$$\left|\sum_i Y_i - \sum_i Y_i^L\right| \le_{\eta_1} 0$$ where $\eta_1=\sum_i \xi_i=\sum_i \Pr[|Y_i|\ge L].$  By Lemma~\ref{lem:coupling_sum_independent2}, we have 
	$$\left|\sum_i Y_i^L - \sum_i \tY_i\right| \le_{\eta_0} \frac{h}{2}\cdot \sqrt{2k\log \frac{2}{\eta_0}} = h \sqrt{\frac{k}{2}\log \frac{2}{\eta_0}}=h_0.$$
	By Lemma~\ref{lem:convolution_wrapping_around_error}, 
	$$\left|\sum_{i=1}^k \tY_i - \lp \tY_1 \oplus_L \tY_2 \oplus_L \cdots \oplus_L \tY_k\rp\right| \le_{\eta_2} 0$$ where $\eta_2 = \Pr\lb\left|\sum_{i=1}^k \tY_i\right|\ge L\rb.$ Finally applying triangle inequality (Lemma~\ref{lem:coupling_triangle_ineq}) once again, we get:
	$$\left|(Y_1+Y_2+\dots+Y_k) - (\tY_1 \oplus_L \tY_2 \oplus_L \dots \oplus_L \tY_k)\right|\le_\eta h_0$$ where $\eta=\eta_0+\eta_1+\eta_2$. We can bound $\Pr\lb\left|\sum_{i=1}^k \tY_i\right|\ge L\rb$ as:
	\begin{align*}
		\Pr\lb\left|\sum_i \tY_i\right|\ge L\rb &= \Pr\lb\left|\sum_i (\tY_i-Y_i^L) + \sum_i (Y_i^L - Y_i)  +\sum_i Y_i \right|\ge L\rb\\
		 &\le \Pr\lb\left|\sum_i (\tY_i-Y_i^L) \right|+\left|\sum_i (Y_i^L - Y_i) \right| +\left|\sum_i Y_i \right|\ge h_0 + 0 + L-h_0\rb\\
		 &\le \Pr\lb\left|\sum_i (\tY_i-Y_i^L) \right| > h_0\rb +\Pr\lb\left|\sum_i (Y_i^L - Y_i) \right| > 0\rb +\Pr\lb\left|\sum_i Y_i \right|\ge L-h_0\rb\\
		 &\le \eta_0 + \eta_1 + \Pr\lb\left|\sum_{i=1}^k Y_i\right|\ge L-h_0\rb.
	\end{align*}
\end{proof}

\begin{proof}[Proof of Theorem~\ref{thm:approximation}]
Combining Corollary~\ref{cor:composePRV_coupling} (with $h_0=\eps_\error$) and Lemma~\ref{lem:coupling_to_privacy_curves}, we have Theorem~\ref{thm:approximation}.	
\end{proof}

\subsection{Tail Bound for PRVs}
\label{sec:appendix_numerical_composition}

To finish the proof of our main theorem (Theorem~\ref{thm:approximation_eps_upper_bound}, we need a tail bound on PRVs in terms of their privacy curves. First, we need a lemma relating the PRVs of a privacy curve $\delta(P||Q)$ with the PRVs of $\delta(Q||P)$.

\begin{proposition}
	\label{prop:finv_prv}
	Let $(X,Y)$ be the PRVs for a privacy curve $\delta(P||Q)$. Then the PRVs for the privacy curve $\delta(Q||P)$ are $(-Y,-X).$
\end{proposition}
\begin{proof}
	Let $(\tX,\tY)$ be the PRVs for $\delta(Q||P)$. We know that $\delta(P||Q)=\delta(X||Y)$. So $\delta(Q||P)=\delta(Y||X).$ Then by Theorem~\ref{thm:PRV_LLRV},
	\begin{align*}
		\tX &= \log\lp \frac{X(t)}{Y(t)}\rp \text{ where } t \sim Y\\
		&= \log\lp e^{-t}\rp \text{ where } t \sim Y\\
		& = -Y.
	\end{align*}
\begin{align*}
		\tY &= \log\lp \frac{X(t)}{Y(t)}\rp \text{ where } t \sim X\\
		&= \log\lp e^{-t}\rp \text{ where } t \sim X\\
		& = -X.
	\end{align*}
\end{proof}

Now, we show our tail bound, which shows the PRVs $(X,Y)$ for a $(\epsilon,\delta)$-DP algorithm satisfies roughly that $\Pr(|Y| \geq \epsilon + 2) \leq 2 \delta$.

\PRVtail*
\begin{proof}
	We have $\delta(X||Y)\le f_{\eps,\delta}$ and $\delta(Y||X)\le f_{\eps,\delta}$ where $f_{\eps,\delta}$ is the privacy curve of a $(\epsilon,\delta)$-DP algorithm. 
	By Theorem~\ref{thm:PRV_LLRV}, we have 
	\begin{align*}
		\delta &\ge \int_0^\infty \Pr[Y\ge \eps+s] e^{-s} ds\\
		       &\ge \int_0^t \Pr[Y\ge \eps+s] e^{-s} ds\\
		       &\ge \Pr[Y\ge \eps+t]\int_0^t e^{-s} ds\\
		       &\ge \Pr[Y\ge \eps+t] (1-e^{-t}).
	\end{align*}

	By Proposition~\ref{prop:finv_prv}, the PRVs for $\delta(Y||X)$ are $(-Y,-X)$. Therefore by a similar argument, we have
	$$\Pr[X \le -\eps-t]=\Pr[-X\ge \eps+t] \le \frac{\delta}{1-e^{-t}}.$$
	Finally, note that $Y(s)=e^s X(s)$ for all $s\in \R$ and $Y(-\infty)=0$ by the definition of PRVs. Therefore $$\Pr[Y\le -\eps-t] \le e^{-\eps-t}\Pr[X\le -\eps-t].$$
	Therefore we have:
	\begin{align*}
			\Pr[|Y|\ge \eps+t] &= \Pr[Y \ge \eps+t] + \Pr[Y \le -\eps-t]\\
			 &\le \Pr[Y \ge \eps+t] + e^{-\eps-t}\Pr[X \le -\eps-t]\\
			 &\le \lp 1 + e^{-\eps-t}\rp\frac{\delta}{1-e^{-t}}.
	\end{align*}
\end{proof}

\subsection{Proof of Theorem~\ref{thm:approximation_eps_upper_bound}}
Now, we can prove our main theorem.
\mainthm*
\begin{proof}
	By Lemma~\ref{lem:PRVs_tailbounds}, 
	\begin{align*}
		\Pr[|Y_i| \ge L] &= \Pr[|Y_i| \ge L-2 + 2]\\
		&\le \delta_{Y_i}(L-2) \cdot \frac{1+e^{-L}}{1-e^{-2}}\\
		&\le \delta_{Y_i}(L-2) \cdot \frac{1+e^{-2}}{1-e^{-2}}\\
		&\le \delta_{Y_i}(L-2) \cdot \frac{4}{3}.
	\end{align*}
	Therefore we have $$\sum_{i=1}^k \Pr[|Y_i|\ge L]\le \frac{4}{3} \sum_{i=1}^k \delta_{Y_i}(L-2) \le \frac{4}{3} \cdot \frac{\delta_\error}{8} = \frac{\delta_\error}{6}.$$
	Similarly $$\Pr[|Y| \ge L-\eps_\error] \le \frac{4}{3}\delta_Y(L-2-\eps_\error)  \le \frac{\delta_\error}{3}.$$ Therefore by Theorem~\ref{thm:approximation}, setting $\eta_0=\frac{\delta_\error}{6}$, we get the desired result.
	
	For the runtime, we note that the bottleneck of our algorithm is to compute the convolution, which can be done using FFT. In total, we need to compute $b+1$ many FFT for $b$ distinct algorithms, one for each for computing the Fourier transform and one of computing the inverse Fourier transform. Since the length of the array for the FFT is bounded by $O(L/h)$, this costs $O(b L/h \log(L/h))$ in total.
	
	The step $\delta_{\tY}(\eps) = \E_{\tY}\lb \lp 1-e^{\eps-\tY}\rp_+\rb$ can be computed in linear time by first computing the CDF of $\tY$ and the prefix sum $\E_{\tY \leq \alpha}\lb e^{-\tY}\rb$ for all $\alpha$.
\end{proof}

\end{document}